\documentclass[11pt,a4paper,english]{amsart}
\usepackage{babel}
\usepackage[latin1]{inputenc}
\usepackage{amsmath}
\usepackage{amsthm}
\usepackage{amsfonts}
\usepackage{mathtools}
\usepackage{indentfirst}
\usepackage{graphicx}
\usepackage{amssymb}
\usepackage[mathscr]{eucal}
\usepackage{tikz}
\usepackage{caption}
\usepackage{hhline,colortbl}
\usepackage{tabularx}

\usepackage{enumerate}
\usepackage{amsthm}
\usepackage{amscd}
\newtheorem{teo}{Theorem}

\newtheorem{pro}[teo]{Proposition}
\newtheorem{cor}[teo]{Corollary}
\newtheorem{lem}[teo]{Lemma}

\theoremstyle{definition}
\newtheorem{rem}[teo]{Remark}
\newtheorem{de}[teo]{Definition}

\DeclarePairedDelimiter\floor{\lfloor}{\rfloor}

\title[Quantum codes defined by trace-depending polynomials]{Stabilizer quantum codes defined by trace-depending polynomials}
\author{Carlos Galindo, Fernando Hernando, Helena Mart\'in-Cruz and Diego Ruano}
\curraddr{\texttt{Carlos Galindo and Fernando Hernando:} Instituto
Universitario de Matem\'aticas y Aplicaciones de Castell\'on andDepartamento de Matem\'aticas, Universitat Jaume I, Campus de Riu
Sec. 12071 Castell\'{o} (Spain)\\
\texttt{Helena Mart\'in-Cruz:} Instituto
Universitario de Matem\'aticas y Aplicaciones de Castell\'on, Universitat Jaume I, Campus de Riu
Sec. 12071 Castell\'{o} (Spain)\\
\texttt{Diego Ruano:} IMUVA-Mathematics Research Institute, Universidad de Valladolid, 47011 Valladolid (Spain).
}
\email{{\rm Galindo:} galindo@uji.es;  {\rm Hernando:} carrillf@uji.es; {\rm Mart\'in:} martinh@uji.es;  {\rm Ruano:} diego.ruano@uva.es}
\date{}
\thanks{Partially funded by MCIN/AEI/10.13039/501100011033 and by ``ERDF A way of making Europe", grants PGC2018-096446-B-C21 and PGC2018-096446-B-C22, by MCIN/AEI/10.13039/501100011033 and by ``ESF Investing in your future", grant RYC-2016-20208, as well as by Universitat Jaume I, grants UJI-B2021-02 and PREDOC/2020/39.}


\begin{document}

\begin{abstract}
Quantum error-correcting codes with good parameters can be constructed by evaluating polynomials at the roots of the polynomial trace \cite{Traza}. In this paper, we propose to evaluate polynomials at the roots of trace-depending polynomials (given by a constant plus the trace of a polynomial) and show that this procedure gives rise to stabilizer quantum error-correcting codes with a wider range of lengths than in \cite{Traza} and with excellent parameters. Namely, we are able to provide new binary records according to \cite{codetables} and non-binary codes improving the ones available in the literature.
\end{abstract}

\maketitle

\section{Introduction}\label{se:uno}
\sloppy The existence of polynomial time algorithms for prime factorization and discrete logarithms on quantum computers  is a clear example that illustrates the importance of quantum computing \cite{22RBC}. Quantum computers are governed by the rules of quantum mechanics since they use subatomic particles to hold memory. Important obstacles for their reliability, like the loss of coherence and the fact that they produce more errors than the classical computers, are solved with quantum error-correcting codes \cite{23RBC, 95kkk}. Thus, quantum error-correction is a key tool in quantum computing, which works despite quantum information cannot be cloned \cite{8AS, 26RBC}. This explains why many researchers are looking for good quantum error-correcting codes.

Seminal papers on quantum error-correcting codes studied binary codes \cite{18kkk, Calderbank, Gottesman} (see also \cite{7kkk, 8kkk, 45kkk}). Later non-binary codes were also considered \cite{AK, Ketkar}; these last codes are particularly interesting in fault-tolerant computation \cite{FTShor, FTKnill, FTPres, FTGot, FTSte, FTCam, luol}. The literature on quantum error-correcting codes is very extensive (some references are \cite{BE,71kkk,lag2,QINP2,CaoCui,Song}). Many of the known quantum error-correcting codes are stabilizer codes. Denote by $\mathbb{C}$ the complex field, let $q$ be a prime power and $n$ a positive integer. A {\it stabilizer code} $\mathcal{C} \neq \{0\}$ is an intersection of eigenspaces (with respect to the eigenvalue $1$) running over the elements of an abelian subgroup of the error group generated by a suitable error basis on the space $\mathbb{C}^{q^n}$. When $\mathcal{C}$ is a $q^k$-dimensional subspace of $\mathbb{C}^{q^n}$ and has minimum distance $d$, we say that $\mathcal{C}$ is an $[[n,k,d]]_q$-code.  A main advantage of stabilizer quantum error-correcting codes is that they can be constructed  from additive codes included in $\mathbb{F}_q^{2n}$ which are self-orthogonal with respect to a trace symplectic form. As particular cases of the above construction, stabilizer codes can be obtained from suitable Hermitian or Euclidean self-orthogonal classical linear codes (see \cite{Ketkar} for details).

Quantum error-correcting codes that achieve the quantum Singleton bound are named quantum MDS codes. One can find many papers on this class of codes (see \cite{Fang, Ball, Liu-LiuX} to cite some articles from the last years). The length of a $q$-ary quantum MDS code is relatively small according to the so-called MDS conjecture \cite{Ketkar}, therefore, fixed a field $\mathbb{F}_q$, it is also desirable to find much longer codes with good parameters. Quantum codes of this last type can be obtained from evaluation codes and their subfield-subcodes \cite{QINP, QINP2, Traza}. The previous references consider large fields, $\mathbb{F}_{q^{2n}}$, and evaluate suitable vector spaces of polynomials with coefficients in $\mathbb{F}_{q^{2n}}$ at suitable roots of the unity where, in addition, one may or may not evaluate at zero. However, in \cite{Traza}, we discovered that evaluating at the set formed by the roots of the trace polynomial $\mathrm{tr}_{2n}(X) = X+ X^q+ \cdots + X^{q^{2n-1}}$, one gets excellent $q$-ary quantum codes, both binary and non-binary.

Motivated by the fact that evaluating at the zeros of the trace polynomial produces codes with good behaviour, we consider trace-depending polynomials, instead of $\mathrm{tr}_{2n}(X)$, with the aim of obtaining quantum codes with new lengths and good parameters. For us, a trace-depending polynomial is a polynomial of the form $a + \mathrm{tr}_{2n}(h(X))$, where $a \in \mathbb{F}_{q^{2n}}$ and $h(X) \in \mathbb{F}_{q^{2n}}[X]$. The benefits of this new procedure are showed at the end of the paper, in Subsection \ref{sporadic}, where \textit{we introduce several trace-depending polynomials} such that evaluating at their roots gives rise to \textit{a considerable number of binary quantum records according to} \cite{codetables}. For us, a (code) record means a binary quantum code such that either its parameters are better than those of a code in \cite{codetables} or match with a missing entry in \cite{codetables}. These codes are stabilizer and we are able to determine their dimensions and minimum distances, but we need to use the computational algebra system Magma \cite{Magma} for checking the Hermitian self-orthogonality of the involved linear codes.

As a consequence, it is undoubtedly interesting to do a theoretical analysis of the family of quantum codes obtained by evaluating at the zeros of trace-depending polynomials. That is, to give conditions guaranteeing self-orthogonality for the constituent linear codes and compute their parameters. Notice that the length of the codes given in \cite{Traza} is $q^{2n-1}$ since $\mathrm{tr}_{2n}(X)$ completely factorizes in the field $\mathbb{F}_{q^{2n}}$, but these new quantum codes have a wider range of lengths. A global study seems untractable because the behaviour of the trace-depending polynomials is unknown.

In this paper, \textit{we restrict ourselves to a specific family of trace-depending polynomials and perform a complete study of the stabilizer quantum codes supported on that family}. This family is formed by the so-called $b$th trace-depending polynomials, $\mathrm{Tr}_b(X)$, where $b=b(t)=1 + q^t$, $0 < t \leq n$  (see Definition \ref{deftraza}); note that, in order to apply our forthcoming  Lemma \ref{newton}, we consider only polynomials $\mathrm{Tr}_b(X)$ which completely factorize in $\mathbb{F}_{q^{2n}}$. Proposition \ref{La5} gives a full description of the polynomials $\mathrm{Tr}_b(X)$, and Theorem \ref{El7} determines when the sum of the $i$th powers, $1\leq i \leq \deg\mathrm{Tr}_b(X)$, of the roots of $\mathrm{Tr}_b(X)$ vanishes, which is a crucial fact for determining the self-orthogonality of the constituent linear codes. This last property is studied in {\it Theorem \ref{era2} giving rise to $q^n$-ary stabilizer quantum codes (see Corollary \ref{el15})}. All these results are presented in Section \ref{se:dos}.

With the above ingredients and using subfield-subcodes, in Subsection \ref{subfield} we {\it determine parameters of $q^{n'}$-ary stabilizer quantum error-correcting codes, where $n'$ divides $n$} (see {\it Theorem \ref{bueno}}). One can also obtain quantum codes by successively using Theorem \ref{era2} and {\it Theorem \ref{Eldearxiv}} (stated in Subsection \ref{Ejemplos}) or Theorems \ref{bueno} and \ref{Eldearxiv}. In this way, we get many new good codes. These codes enlarge the constellation of lengths of the quantum error-correcting codes obtained by evaluating at the zeros of the trace polynomial \cite{Traza}.

Subsections \ref{binary} and \ref{nonbinary} supply quantum codes constructed with our theoretical results. In Subsection \ref{binary}, we prove that \textit{our development gives rise to new and good binary quantum codes, some of them being records according to} \cite{codetables}. In Subsection \ref{nonbinary} we provide \textit{new non-binary quantum error-correcting codes, some of them improving the parameters of the codes available in the literature}. All the given codes have parameters exceeding the quantum Gilbert-Varshamov bounds \cite{FengMa, Ketkar, 71kkk}.

\section{Evaluation codes and $b$th trace-depending polynomials}
\label{se:dos}


In this section, we introduce a particular family of trace-depending polynomials and consider linear codes that evaluate at the roots of the polynomials in this family. We study their parameters and self-orthogonality conditions. Later, we will see that good stabilizer quantum codes can be derived from them and their subfield-subcodes.

\subsection{The $b$th trace-depending polynomials}
\label{sb:dosuno}

Let $q$ be a prime power. Since in the future we will be interested in subfield-subcodes and Hermitian duality, our initial results are stated over the field $\mathbb{F}_{q^{2n}}$ with $n$ a positive integer.

For defining the trace-depending polynomials we are interested in, we consider the trace polynomials $\mathrm{tr}_{2n}(X)$ and $\mathrm{tr}_{n}(X)$ defined as follows:
\[
\mathrm{tr}_{j}(X) : = X+ X^q+ X^{q^2}+ \cdots + X^{q^{j-1}},
\]
where $j$ equals either $2n$ or $n$. Next, set $b = b(t)=1 + q^t$ for some integer number $0< t \leq n$ and introduce the polynomial
\begin{equation*}
P_b(X):=\left\{
\begin{array}{lcc}
1 + \mathrm{tr}_{2n}(X^b) & \mathrm{if} & 0 < t< n \\
1 + \mathrm{tr}_{n}(X^b) & \mathrm{otherwise} & (t=n).
\end{array}
\right.
\end{equation*}
Now, consider the quotient ring $R:= \mathbb{F}_{q^{2n}}[X] /\langle X^{q^{2n}-1} -1  \rangle$ and we are ready for introducing the concept of $b$th trace-depending polynomial.

\begin{de}
\label{deftraza}
{\rm With the above notation, we denote by $\mathrm{Tr}_b(X)$ the representative with minimum degree of the class of $P_b(X)$ in $R$. We name $\mathrm{Tr}_b(X)$ the $b$th trace-depending polynomial. Furthermore, it defines a polynomial map that we call the $b$th trace-depending map.
}
\end{de}

\begin{rem}
{\rm
Later we will introduce codes obtained by evaluating at the roots of the  polynomial $\mathrm{Tr}_b(X)$. When $t=n$, Definition \ref{deftraza} uses the polynomial $\mathrm{tr}_{n}(X)$ instead of $\mathrm{tr}_{2n}(X)$ because, when the characteristic of the field $\mathbb{F}_{q^{2n}}$ is $2$, otherwise $\mathrm{Tr}_b(X) = 1 $ and  $\mathrm{Tr}_b(X)$ has no roots. Indeed, when $t=n$, a simple computation shows the equality
\[
\mathrm{tr}_{2n}(X^b) + \langle X^{q^{2n}-1} -1  \rangle = 2 \left[ \mathrm{tr}_{n}(X^b) + \langle X^{q^{2n}-1} -1  \rangle \right],
\]
which also proves that, in this case ($t=n$), when the characteristic of the field $\mathbb{F}_{q^{2n}}$ is not even, the trace maps $\mathrm{tr}_{2n}$ and $\mathrm{tr}_{n}$ play an analogous role.
}
\end{rem}

Next, we determine the degree of the polynomial $\mathrm{Tr}_b(X)$.
\begin{pro}
\label{grado}
The degree of the $b$th trace-depending polynomial $\mathrm{Tr}_b(X)$, $b=1 +q^t, \; 1 < t \leq n$, is $m=m(t) = q^{2n-1-t} + q^{2n-1}$.
\end{pro}
\begin{proof}
The case when $t=n$ is clear. Assume $0 <t <n$ and write
$$
\mathrm{Tr}_b(X) = \sum_{k=0}^m a_k X^k.
$$
Since $P_b(X)$ has no term involving a power of $X^{q^{2n}-1}$ and $\mathrm{Tr}_b(X)$ is the representative with minimum degree of its class in $R$, $a_0=1$. Write $k=\sum_{\ell=0}^{2n-1} \kappa_\ell q^\ell$, with $0 \leq \kappa_\ell < q$, the $q$-adic expansion of the exponents $k>0$ such that $a_k \neq0$. Sometimes, for the sake of simplicity and easiness, this $q$-adic expansion will be represented with a $2n$-tuple $(k)_q$ as in Table \ref{QA1}.
\begin{table}[ht]
\centering
\begin{tabular}{|c|c|c|c|c|c|c|c|c|c|}
  \hline
   & $q^{0}$ & $q^{1}$ & $\cdots$ & $q^{t-1}$ & $q^{t}$ & $\cdots$ & $q^{2n-t}$ & $\cdots$ & $q^{2n-1}$ \\ \hline
   $(k)_q$ & $\kappa_0$ &  $\kappa_1$ & $\ldots$ &  $\kappa_{t-1}$ &  $\kappa_{t}$ & $\ldots$ &  $\kappa_{2n-t}$ & $\ldots$ &  $\kappa_{2n-1}$\\
  \hline
\end{tabular}
\caption{$q$-adic expansion of $k$}
\label{QA1}
\end{table}
The $q$-adic expansion of $b$ is displayed in Table \ref{QA2},
\begin{table}[ht]
\centering
\begin{tabular}{|c|c|c|c|c|c|c|c|c|c|c|}
  \hline
   & $q^{0}$ & $q^{1}$ & $\cdots$ & $q^{t-1}$ & $q^{t}$ & $q^{t+1}$ & $\cdots$ & $q^{2n-t}$ & $\cdots$ & $q^{2n-1}$ \\ \hline
  $(b)_q$ & $1$ &  $0$ & $\ldots$ &  $0$ &  $1$ & $0$ & $\ldots$ &  $0$ & $\ldots$ &  $0$\\
  \hline
\end{tabular}
\caption{$q$-adic expansion of $b$}
\label{QA2}
\end{table}
and the $q$-adic expansion of the elements $k>0$ such that $a_k \neq0$ can be obtained by successively shifting the values in Table \ref{QA2}. Indeed, each shift corresponds to an exponent $k=q^s b$, $0 \leq s \leq 2n-1$, where $q^{2n}$ is identified with $1$. As a consequence, $a_k=1$ whenever $a_k \neq0$, and the degree of the $b$th trace-depending polynomial $\mathrm{Tr}_b(X)$ is given by the sequence of shifts which gives the largest positive integer; it is $m= q^{2n-1-t} + q^{2n-1}$. Notice that, for simplicity, $b=q^0 b$ is considered a shift of $b$.
\end{proof}

Along this paper and with the above notation, we only consider triples $(q,n,b)$ satisfying the following property:
\begin{equation}
\label{A}
\mbox{ The polynomial $\mathrm{Tr}_b(X)$ has $m$ different roots in the  field $\mathbb{F}_{q^{2n}}$.}
\end{equation}
We denote these roots by $\{\beta_1, \beta_2, \ldots, \beta_m\}$.

The following result proves that the triples $(q,n,b(n))$ always satisfy Property (\ref{A}). For $0 < t< n$, by explicit computation, we have found a number of triples $(q,n,b)$ satisfying Property (\ref{A}). Some examples can be seen in Table \ref{TTabla1}. With some of these values we have obtained good stabilizer quantum codes. We do not know a general result characterizing the before mentioned triples.

\begin{table}[ht]
\centering
\begin{tabular}{||c|c|c|c||c||}
  \hline \hline
 $q$ & $n$ & $t$  & $b$ & $m$=degree    \\
 \hline \hline
    2& 2& 1 & 3 &12 \\
2& 4& 2 & 5 &160 \\
2& 4& 3 & 9 &144 \\
2& 6& 3 & 9 &2304 \\
2& 6& 5 & 33 &2112 \\
3& 2& 1 & 4 &36 \\
3& 4& 2 & 10 &2430 \\
3& 4& 3 & 28 &2268 \\
5& 2& 1 & 6 &150 \\
5& 4& 2 & 26 &81250 \\
5& 4& 3 & 126 &78750 \\
7& 2& 1 & 8 &392 \\
11& 2& 1 & 12 &1452 \\
\hline
 \hline
\end{tabular}
\caption{Triples $(q,n,b)$, $b = 1 + q^t$, satisfying Property (\ref{A})}
\label{TTabla1}
\end{table}

\begin{pro}
\label{La2}
Assume $b:= b(n)$. The $b$th trace-depending polynomial $\mathrm{Tr}_b(X) \in \mathbb{F}_{q^{2n}}[X]$  has $m = m(n):=q^{n-1}+q^{2n-1}$ different roots in the field $\mathbb{F}_{q^{2n}}$.
\end{pro}
\begin{proof}
The polynomial $\mathrm{tr}_{n}(X)$ gives the trace map $\mathrm{tr}_{n}: \mathbb{F}_{q^n} \rightarrow \mathbb{F}_q$. The map $g: \mathbb{F}_{q^{2n}} \rightarrow \mathbb{F}_{q^n}$ defined as $g(x)=x^b$ is well-defined and it is surjective; with the exception of $0 \in \mathbb{F}_{q^n}$, each element in $\mathbb{F}_{q^n}$ has $q^n +1$ counter-images. The map $P_b - 1$ defined by $P_b(X) - 1$ satisfies $P_b - 1 = \mathrm{tr}_{n} \circ g$, therefore the roots of $\mathrm{Tr}_b(X)$ are exactly the set $$(\mathrm{tr}_{n} \circ g)^{-1} (-1)= g^{-1}[\mathrm{tr}_{n}^{-1}(-1)].$$ Since $\mathrm{tr}_{n}$ is a trace map, $\mathrm{tr}_{n}^{-1}(-1)$ has $q^{n-1}$ different elements and the cardinality of $(\mathrm{tr}_{n} \circ g)^{-1} (-1)$ is $(q^n +1)q^{n-1} = q^{2n-1} + q^{n-1}$, which concludes the proof.
\end{proof}

\subsection{Evaluation codes}
\label{sb:dosdos}
Now we are going to define the family of codes we are interested in. We only consider triples $(q,n,b)$ satisfying Property (\ref{A}). We fix any of them, set $\mathrm{Tr}(X) := \mathrm{Tr}_b(X)$ and  define the evaluation map $\mathrm{ev}_{\mathrm{Tr}}$ at the roots of $\mathrm{Tr}(X)$, $\{\beta_1, \beta_2, \ldots, \beta_m\}$, as:
\begin{equation}
\label{evaluation}
\mathrm{ev}_{\mathrm{Tr}}: \mathbb{F}_{q^{2n}}[X]/\langle \mathrm{Tr}(X) \rangle \to \mathbb{F}_{q^{2n}}^m, \mbox{ given by $\mathrm{ev}(h)=(h(\beta_1),h(\beta_2), \ldots, h(\beta_m))$,}
\end{equation}
where $h$ stands for the class of a polynomial $h \in \mathbb{F}_{q^{2n}}[X]$ in $\mathbb{F}_{q^{2n}}[X]/\langle \mathrm{Tr}(X) \rangle$ and its corresponding polynomial function.

\begin{de}
{\rm
Let $\mathcal{H}=\{0,1, \ldots, m-1\}$ and consider a non-empty subset $\Delta \subseteq \mathcal{H}$.  We define the  evaluation code, $E_{\Delta, \mathrm{Tr}}$, of $\Delta$ at the roots of the trace-depending polynomial $\mathrm{Tr}(X)$ (given by a triple $(q,n,b)$) as the linear code of length $m$ over the field $\mathbb{F}_{q^{2n}}$ generated by the set $\{\mathrm{ev}_{\mathrm{Tr}} (X^i)\;|\; i \in \Delta\}$.}
\end{de}

Our next result describes the polynomial $\mathrm{Tr}(X)$. Recall that $b=1+ q^t$ with $0<t\leq n$.

\begin{pro}
\label{La5}
Let $\mathrm{Tr}(X) = \sum_{k=0}^{m} a_k X^k$. One has that $a_k=0$ for all indices $k$, with the exception of:
\begin{itemize}
\item $k=0$;
\item $k=q^j b$, where $0 \leq j \leq 2n -t -1$;
\item and, when $t <n$, $k=k_j : = q^{j-1} (1+q^{2n-t})$ for  $1 \leq j \leq t$.
\end{itemize}
Thus, $\mathrm{Tr}(X)$ has $2n+1$  non-zero coefficients $a_k$ when $t <n$ and it has $n+1$ otherwise ($t=n$). All the non-vanishing coefficients are equal to $1$.
\end{pro}
\begin{proof}
It is clear that $a_0=1$. The monomials in the second item of the statement: $X^{q^j b}$, $0 \leq j \leq 2n -t -1$, are terms with coefficient $1$ in the polynomial $\mathrm{Tr}(X)$ by the construction of $P_b(X)$, and they are the only terms with non-vanishing coefficient when $t=n$ because taking classes modulo the ideal $\langle X^{q^{2n}-1} -1 \rangle$ does not produce any modification of $P_b(X)$.

When $t<n$, apart from the above monomials, there are new terms with coefficient $1$ in the expression of  $P_b(X)$ which are $X^{q^j b}$, $2n -t \leq j \leq 2n -1$. Recall that $\mathrm{Tr}(X)$ is the representative of minimum degree of the class of $P_b(X)$ modulo $\langle X^{q^{2n}-1} -1 \rangle$. As we explained in the proof of Proposition \ref{grado}, the representatives of the classes modulo $\langle X^{q^{2n}-1} -1 \rangle$ of the monomials $X^{q^j b}$, $0 \leq j \leq 2n-1$, are monomials $X^k$ where $k$ is an integer whose $q$-adic expansion (see Table \ref{QA1}) is given by a sequence of shifts of the $q$-adic expansion of $b$ (see Table \ref{QA2}). Clearly the monomials $X^{q^j b}$, with $0 \leq j \leq 2n -t -1$, correspond to the first shifts and those where $2n -t \leq j \leq 2n -1$ correspond to the values $k_j$ in the last item of the statement.
\end{proof}

\begin{rem}
Table \ref{TTableoplus} shows the $q$-adic expansions of the indices $k \neq 0$ in the expression of $\mathrm{Tr}(X) = \sum_{k=0}^{m} a_k X^k$ such that $a_k \neq 0$. Recall that the values $k_j$ introduced in Proposition \ref{La5} do not appear when $t=n$ and notice that $b$ and $k_1$ are the unique indices which are not a multiple of $q$. The $q$-adic expansions show how the indices $k$ are ordered as natural numbers, for instance $m > k_t$ and both are larger than the remaining ones.
\end{rem}

\begin{table}[ht]
\centering
\begin{tabular}{|c|c|c|c|c|c|c|c|c|c|c|c|}
  \hline
   & $q^{0}$ & $q^{1}$ & $\cdots$ & $q^{t-1}$ & $q^{t}$ & $q^{t+1}$ & $\cdots$ & $q^{2n-t-1}$ &$q^{2n-t}$  & $\cdots$ & $q^{2n-1}$ \\ \hline
  $(b)_q$ & $1$ &  $0$ & $\ldots$ &  $0$ &  $1$ & $0$ & $\ldots$ & $0$& $0$ & $\ldots$ &  $0$\\
  \hline
  $(qb)_q$ & $0$ &  $1$ & $\ldots$ &  $0$ &  $0$ & $1$ & $\ldots$ & $0$& $0$ & $\ldots$ &  $0$\\\hline
  $\vdots$ & $\vdots$ & $\vdots$ & $\cdots$ &  $\vdots$ &  $\vdots$ & $\vdots$ & $\cdots$ & $\vdots$& $\vdots$&  $\cdots$ &  $\vdots$\\ \hline
  $(m)_q=(q^{2n-t-1}b)_q$ & $0$ &  $0$ & $\ldots$ &  $0$ &  $0$ & $0$ & $\ldots$ & $1$& $0$ & $\ldots$ &  $1$\\ \hline
  $(k_1)_q$ & $1$ &  $0$ & $\ldots$ &  $0$ &  $0$ & $0$ & $\ldots$ & $0$& $1$ & $\ldots$ &  $0$\\ \hline
  $\vdots$ & $\vdots$ & $\vdots$ & $\cdots$ &  $\vdots$ &  $\vdots$ & $\vdots$ & $\cdots$ & $\vdots$& $\vdots$&  $\cdots$ &  $\vdots$\\
  \hline
  $(k_t)_q$ & $0$ &  $0$ & $\ldots$ &  $1$ &  $0$ & $0$ & $\ldots$ & $0$& $0$ & $\ldots$ &  $1$\\ \hline
\end{tabular}
\caption{$q$-adic expansions of the indices $k\neq 0$ such that $a_k \neq 0$}
\label{TTableoplus}
\end{table}

We are interested in codes $E_{\Delta,{\mathrm{Tr}}} \subseteq \mathbb{F}_{q^{2n}}^{m}$ which are self-orthogonal with respect to the  Hermitian inner product because we aim to construct quantum stabilizer codes. For this reason (see the proof of the forthcoming Theorem \ref{era2}),  we introduce the following values:
\[
s_i := \sum_{j=1}^{m} \beta_j^i;\;\;\; 1 \leq i \leq q^{2n}-1.
\]

Now we state a  result involving the above values $s_i$ in the fashion of \cite[Lemma 4]{Traza}, which can be proved similarly.

\begin{lem}
\label{newton}
With the above notation, for every index $r$ such that $1 \leq r \leq m$, the following equality
\[
\left( \sum_{j=0}^{r-1} a_{m-j} s_{r-j} \right) + r a_{m-r}=0
\]
holds.

In addition, when $r > m$, one gets
\[
\sum_{j=0}^{m} a_{m-j} s_{r-j} = 0.
\]
\end{lem}

The following result determines the indices $i \leq m$ for which the value $s_i$ does not vanish and, therefore, it helps to show when Hermitian orthogonality of vectors $\mathrm{ev}_{\mathrm{Tr}} (X^s)$ does not hold (see, again, the proof of Theorem \ref{era2}). For that purpose, consider the set of indices:
\[
j_{2, \ell}:= 1 + (2+ \ell)q^{t-1} + (q-(2+ \ell)) q^{2n-t-1},
\]
where $0 \leq \ell \leq q-1$.
\begin{teo}
\label{El7}
Keep the notation as in Proposition \ref{La5}.

-- When $1 < t <n$, there are exactly $q$ indices $i \leq m$ such that $s_i \neq 0$. We denote these indices in increasing order as $i_0 < i_1$ whenever $q=2$. Otherwise, we denote them as
\[
i_0 < i_1 < i_{2,0} < i_{2,1} < \cdots < i_{2,q-3}.
\]
Then,
\begin{description}
\item[a)] It holds that $i_0=m-k_1$, $i_1= m-(k_1+k_t -m)$, and $i_{2,\ell}= m- j_{2, \ell}$ for $0 \leq \ell \leq q-3$.
\item[b)] $s_{i_0}= s_{i_{2,\ell}} = 1$ for $\ell=0$ and for $\ell$ even.
\item[c)] $s_{i_1}= s_{i_{2,\ell}} = -1$ for $\ell$ odd.
\end{description}
\vspace{1mm}

-- When $t=1$, there exist exactly $q+1$ indices $i \leq m $ such that $s_i \neq 0$. With the above notation, these indices are
\[
i_0 < i_1 < i_{2,0} < i_{2,1} < \cdots < i_{2,q-3} < i_{2,q-2}:=i_{2,q-3} + q^{2n-2}-q^{2n-3}+q-1,
\]
and they satisfy Property a), and properties b) and c) for $0 \leq \ell \leq q-2$. The case $q=2$, $n=2$ and $t=1$ should be treated separately; here there are four indices $i_0 < i_1 < i_{2,0} < i_{2,1}$ satisfying the above mentioned properties a), b) and c).

-- When $t=n$, there is only one index $i_0= q^{2n-1} - q^n + q^{n-1} -1$ lower than $m$ such that $s_{i_0} \neq 0$. Here we also find an exception in the case $q=t=n=2$ where there are two indices $i_0=5$ and $i_1=10$ satisfying $s_i \neq 0$.

\end{teo}

\begin{proof}
Lemma \ref{newton} and $q$-adic expansions are the main tools of our proof. We divide it in two cases: Case A, where we study the case $t\neq 1$ and Case B that corresponds to the situation $t=1$. Within each case, we consider several steps and state and prove some lemmas. Step A.1 computes $i_0$ and $s_{i_0}$ proving the first equalities in a) and b) and also the first statement in the case $t=n$. Step A.2 determines $i_1$ and $s_{i_1}$ showing the second equality in a) and the first one in c). Here we also conclude the proof of the case $t=n$.
Step A.3 (respectively, A.4) computes $i_{2,0}$ and $s_{i_{2,0}}$ (respectively, $i_{2,\ell}$ and $s_{i_{2,\ell}}$ for $\ell \neq 0$). The treatment of Case B is a bit different because distinct $q$-adic expansions appear. We consider here three steps corresponding to results which will prove the statement of Theorem \ref{El7} in this case $t=1$.

Our strategy mainly consists of noticing that fixed an index $h$ such that $s_h \neq 0$, the next index $h' >h$ such that $s_{h'} \neq 0$ occurs when the coefficient of $s_h$ in the sum provided by Lemma \ref{newton} that starts with $a_m s_{h'}$ is different from zero. To reach this conclusion, we also prove that the mentioned coefficient is zero for those indices $h^{\prime \prime}$ such that $h < h^{\prime \prime} < h'$ and, in this case, $s_{h^{\prime \prime}} =0$.\\

{\it Case A. $t\neq 1$.} {\it Step A.1.} We start by proving the first equality in items a) and b) of the statement. Thus, we assume $1 < t <n$. As before, we set $\mathrm{Tr}_b(X) = \sum_{k=0}^m a_k X^k$. Denoting $\mathrm{Supp}_{\mathrm{Tr}_b(X)} = \mathrm{Supp}:=\{k \neq 0 \;|\; a_k \neq 0\}$, we have proved that
\[
\mathrm{Supp} = \{ k \;|\; (k)_q \mbox{ is obtained by iteratively applying shifts to $(b)_q$}  \}
\]
and this set has cardinality $2n$.

From our notation $i_0 := \min \{ i  \; | \; 1 \leq i \leq m \mbox{ and $s_i \neq 0$}\}$. Setting $r=i_0$ in the first equality of Lemma \ref{newton}, we get
\[
a_m s_{i_0}+ a_{m-1} s_{i_0-1} + \cdots + a_{m-(i_0-1)} s_{1} = -i_0 a_{m-i_0}.
\]
The definition of $i_0$ shows that $0 \neq s_{i_0} =  -i_0 a_{m-i_0}$. Since our ground field is $\mathbb{F}_{q^{2n}}$, this implies that $-i_0$ is not a multiple of $q$ and $a_{m-i_0} \neq 0$. Taking into account that $m=q^{2n-1-t} + q^{2n-1}$, $m-i_0$ has to be of the form $1 + \gamma q$ because otherwise $i_0$ would be a multiple of $q$ and thus $i_0 a_{m-i_0}$ would be $0$. The index $i_0$ is a minimum and therefore $m-i_0$ equals the value
\[
\max \{ m- k=m - \sum_{\ell=0}^{2n-1} \kappa_\ell q^\ell\;|\;  m - k \in \mathrm{Supp} \; \mbox{and its $q$-adic expansion starts with $1$} \}.
\]
Inspecting the set of $q$-adic expansions that are obtained as (successive) shifts of $(b)_q$ -see Table \ref{TTableoplus}- one deduces that, with the notation in Proposition  \ref{La5}, the following equality holds:
\[
m-i_0=k_1 = 1 + q^{2n-t}
\]
and hence $s_{i_0}= -(-1) =1$. Therefore the first equality in items a) and b) of the statement follows  for $1 < t<n$.

{\it When $t=n$,} noticing that the cardinality of $\mathrm{Supp}$ is $n$ and reasoning analogously, we obtain $m-i_0 = b = 1+q^n$ and therefore $i_0 = q^{2n-1} - q^n + q^{n-1} -1$. This proves our last statement with the exception of the uniqueness of $i_0$ that will be proved later.\\

{\it Step A.2.} Let us prove the first equality in Item c) of the statement. Assume  $1 < t < n$ and, as in the statement of the theorem, set $$i_1:= \min \{i \;|\; i_0 < i \leq m \mbox{ and $s_i \neq 0$} \}.$$ Again by Lemma \ref{newton}, one gets:
\[
a_m s_{i_1}+ a_{m-1} s_{i_i-1} + \cdots + a_{m+i_0-i_1} s_{i_0} + a_{m+i_0-i_1-1} s_{i_0-1} + \cdots + a_{m-(i_1-1)} s_{1} = -i_1 a_{m-i_1},
\]
which, from the definition of $i_1$, implies
\begin{equation}
\label{twostars}
s_{i_1} + a_{m+i_0-i_1} s_{i_0} =  -i_1 a_{m-i_1}.
\end{equation}
The inequalities $i_0 < i \leq m$ prove $0 \leq  m-i < m- i_0 =k_1$ and thus $i_0 \leq m +i_0 - i< i_0 + k_1$. We look for $i_1$ such that $a_{m+i_0-i_1} \neq 0$ (later we will see that this is the only possibility) and then $m+i_0-i_1$ must be equal to the value
$$\max \{m+i_0-i \in \mathrm{Supp} \;|\; i_0 \leq m+ i_0 -i < m\}.$$
Considering the $q$-adic expansions of the values in $\mathrm{Supp}$ (see Table \ref{TTableoplus}), that maximum is attained when $m+i_0-i_1 = k_t$, because $k_t$ is the larger value in $\mathrm{Supp}$ lower than $m$. Then, $i_1 =m + i_0 - k_t = m - (k_1 + k_t -m)$ as stated in a). The index $i_1$ is a multiple of $q$ if and only if $m-i_1$ is, thus looking at the $q$-adic expansions of the shifts of $(b)_q$, $i_1 a_{m-i_1}$ equals $0$ except when $m-i_1$ equals $b$ or $k_1$. Now $m-i_1$ is neither $b$ nor $k_1$ and therefore, by Equality (\ref{twostars}), $s_{i_1} = -1$ as said in Item c) of our theorem. Indeed, reasoning by contradiction, if $m-i_1 =b$  then $k_t -i_0 = b$ and thus $k_t + k_1 = m+b$, which means $q^{t-1} + q^{2n-t} = q^{t} + q^{2n-t-1}$, a contradiction. In addition, $m-i_1=k_1$ implies $k_t-i_0 =k_1$ and therefore $k_t=m$, again a contradiction.

Notice that in the searching of $i_1$, one could consider Equality (\ref{twostars}) with some index $i$, $i_0 < i < i_1$ instead of $i_1$, $s_i \neq 0$ and $a_{m+i_0-i}=0$, but then $m -i$ should be either $k_1$ or $b$. In the first case $i=m-k_1=i_0$ which contradicts the fact that $i_0 < i$; and in the second one, the inequality $q^{2n-1} + q^{2n-t-1} + 1 + q^t < q^{2n-1} + q^{t-1} + 1 + q^{2n-t}$ proves $m+b < k_t + k_1$, which implies $i_1 = 2m -k_1 - k_t < m-b=i$ and $i$ would not be the required minimum value with $s_i \neq 0$.

{\it Now we conclude the proof of our last statement concerning the case $t=n$.} Reasoning as in the previous paragraphs, one gets two possibilities.

The first one is that Equality (\ref{twostars}) holds for some index $i_1$ such that $a_{m+i_0-i_1}=0$, then $m-i_1=b$ and, since we proved before that in this case $m-i_0=b$, then $i_1=i_0$, which contradicts the fact $i_1 > i_0$.

Otherwise, $m+i_0 -i_1$ should be an element in $\mathrm{Supp}$ of the form $(1+q^n) q^{j}$, for some $0 < j \leq n-2$, because Equality (\ref{twostars}) makes no sense for $j=0$ -except when $q=n=2$- nor for $j=n-1$ (since it would imply that $i_0=i_1$). Then
\[
i_1 = (1+q^n)q^{n-1} +(1+q^n)q^{n-1} - (1+q^n) - (1+q^n) q^{j} = (1+q^n) (2 q^{n-1} -q^{j} -1),
\]
and thus $i_1>m$, proving that there is no such $i_1\leq m$. As a consequence, we conclude that $i_0$ is the only index satisfying $s_{i_0} \neq 0$ when $t=n$ and the case $q=n=2$ does not hold.

Notice that $m+i_0-i_1=1+q^n=b$ if and only if $m+i_0-i_1=m-i_0$, which is equivalent to $i_1=2i_0$ and then $2i_0 \leq m$. This inequality happens if and only if $q^{2n-1} + q^{n-1} \leq 2 q^n +2$, which holds only when $q=n=2$. Therefore, only in this case, we get a new index $i_1$ such that $s_{i_1}=-1=1$ as stated.\\

{\it Step A.3.} Assume $ 1 < t<n$. Iterating our reasoning, define
$$i_2 := \min \{ i  \; | \; i_1 < i \leq m \mbox{ and $s_i \neq 0$}\}.$$ By Lemma \ref{newton} one gets
\begin{equation}
\label{threestars}
a_m s_{i_2}+  \cdots + a_{m+i_1-i_2} s_{i_1} +  \cdots + a_{m+i_0-i_2} s_{i_0}+ \cdots + a_{m-(i_2-1)} s_{1} = -i_2 a_{m-i_2},
\end{equation}
where the main novelty is that one might have three non-vanishing summands on the left hand side of the equality.

Let us study Equality (\ref{threestars}). First we determine the $q$-adic expansion of the value $i_1$.

\begin{lem}
\label{lemma1}
With the above notation, the $q$-adic expansion of $i_1$ is that displayed in Table \ref{QA3}.
\begin{table}[ht]
\centering
\begin{tabular}{|c|c|c|c|c|c|c|c|c|c|c|c|c|}
  \hline
   & $q^{0}$ & $\cdots$ & $q^{t-2}$ & $q^{t-1}$ & $q^{t}$ &  $\cdots$ & $q^{2n-t-2}$ & $q^{2n-t-1}$ & $q^{2n-t}$ & $\cdots$ & $q^{2n-2}$ & $q^{2n-1}$ \\ \hline
  $(i_1)_q$ & $q-1$ & $\ldots$ & $q-1$ & $q-2$ &  $q-1$ & $\ldots$ &  $q-1$ & $1$  &$q-1$ &$\ldots$ & $q-1$ &  $0$\\
  \hline
\end{tabular}
\caption{$q$-adic expansion of $i_1$}
\label{QA3}
\end{table}
\end{lem}

\begin{proof}
We start with the following chain of equalities
\begin{multline}
\label{dellemma}
i_1= (m-k_1)+(m-k_t) = q^{2n-1} + q^{2n-1-t} -\left(1 + q^{2n-t}\right) + q^{2n-1} + q^{2n-1-t} -\left(q^{2n-1}+ q^{t-1}\right)\\
=q^{2n-1}+ q^{2n-1-t} - \left(1 + q^{2n-t}\right)+ q^{2n-1-t} - q^{t-1}:=w.
\end{multline}
Noticing that $q^{2n-1} = (q-1)q^{2n-2} + (q-1)q^{2n-3} + \cdots + (q-1)q + q$, one gets that the value $w$ in (\ref{dellemma}) equals
\begin{multline*}
(q-1)q^{2n-2}+ \cdots + (q-2)q^{2n-t} + (q+1)q^{2n-t-1} + (q-1)q^{2n-t-2} + \cdots \\ +(q-1)q^{t} + (q-2)q^{t-1} + \cdots + (q-1),
\end{multline*}
which ends the proof.
\end{proof}

Next we study the index $i_2$ involved in Equality (\ref{threestars}).
\begin{lem}
\label{eralema2}
There is only an index $i_2^{'}>i_1$ such that $a_{m+i_1-i_2^{'}} \neq 0$. With the notation as before Theorem \ref{El7}, this index satisfies $m- i_2^{'} = j_{2,0}$.
\end{lem}
\begin{proof}
Since $i_2^{'}>i_1$, there exists a positive integer $j^{'}_2<k_1+k_t-m$ such that $i_2^{'} = m - j^{'}_2$. Then $m + i_1- i_2^{'} = i_1 + j^{'}_2<m$. By Lemma \ref{lemma1}, $k=k_t$ is the unique value $k= m + i_1- i_2^{'}<m$ as in Proposition \ref{La5} that can be obtained with indices $j^{'}_2<k_1+k_t-m$. This is because the last coordinate of the $q$-adic expansion $(k)_q$ of the remaining values $k$ in Proposition \ref{La5} vanishes. By inspection, we deduce that the $q$-adic expansion of $j^{'}_2$ is that given in Table \ref{QA4} and thus, with the notation as before the statement of Theorem \ref{El7}, $j^{'}_2 = j_{2,0}$.
\begin{table}[ht]
\centering
\begin{tabular}{|c|c|c|c|c|c|c|c|c|c|c|c|c|}
  \hline
   & $q^{0}$ & $q^{1}$ & $\cdots$ & $q^{t-2}$& $q^{t-1}$ & $q^{t}$ & $\cdots$ & $q^{2n-t-2}$ & $q^{2n-t-1}$ & $q^{2n-t}$& $\cdots$ & $q^{2n-1}$ \\ \hline
  $(j^{'}_2)_q$ & $1$ &  $0$ & $\ldots$ & $0$ & $2$ &  $0$ &  $\ldots$ & $0$ & $q-2$ & $0$ & $\ldots$ &  $0$\\
  \hline
\end{tabular}
\caption{$q$-adic expansion of $j^{'}_2$}
\label{QA4}
\end{table}
\end{proof}

\begin{lem}
\label{elonce}
The above introduced index $i_2 = \min \{ i  \; | \; i_1 < i \leq m \mbox{ and $s_i \neq 0$}\}$ equals $m- j_{2,0} := i_{2,0} $. In addition, $s_{i_{2,0}} = 1$.
\end{lem}
\begin{proof}
Assume first that $q>2$. Define $j_2=m-i_2$, then
\begin{equation}
\label{dejota2}
j_2 < m-i_1=k_1 + k_t - m = 1 + (q-1) q^{2n-1-t} + q^{t-1}:=\theta.
\end{equation}
The $q$-adic expansion of $\theta$ in the above equality proves that there is no index $j$, $j_{2,0} < j < k_1+k_t -m$ such that $s_{m-j} \neq 0$. In fact, write $i =m -j$, by Lemma \ref{eralema2}, we have that $a_{m+i_1-i} = 0$; in addition $ia_{m-i} =0$ because either $i$ is a multiple of $q$ or otherwise $a_{m-i}=0$. This is because inspecting the $q$-adic expansion of $j_{2,0}$ and the expression (\ref{dejota2}), we notice that there is no $j$ as required with a $q$-adic expansion having only two ones; in fact $j$ should have a $q$-adic expansion of the form $1 + (q-2)q^{2n-t-1} + \mbox {other terms}$. Finally, $a_{m+i_0-i}=a_{m-k_1+j}=0$ since all the coefficients of the $q$-adic expansion of $m-k_1$ are $q-1$ except those of $q^{2n-1}$ and $q^{2n-1-t}$. Therefore, considering Equality (\ref{threestars}) with $i$ instead of $i_2$, we get $s_{m-j} = 0$. As a consequence, $i_2 = m - j_{2,0} := i_{2,0}$ is our candidate for satisfying $s_{i_2} \neq 0$. Let us show that, indeed, $s_{i_{2,0}} \neq 0$.

Equality (\ref{threestars}) reads
\begin{equation*}
 s_{m- j_{2,0}}+  \cdots + a_{k_t} s_{i_1} +  \cdots + a_{m-k_1+j_{2,0}} s_{i_0} = -(m-j_{2,0}) a_{j_{2,0}}
\end{equation*}
and we know that $a_{k_t}=1$,  $s_{i_1}= -1$. Now,
\begin{multline}
m-k_1+j_{2,0} = q^{2n-1} + q^{2n-1-t} - q^{2n-t} -1 + (q-2)q^{2n-1-t} + 2 q^{t-1} +1  \\
= q^{2n-1} + q^{2n-1-t}  -q^{2n-t} -1  + q^{2n-t} - 2 q^{2n-t-1} + 2 q^{t-1} +1 =  q^{2n-1}  -q^{2n-t-1} + 2 q^{t-1} \\
= (q-1) q^{2n-2} + (q-1) q^{2n-3} + \cdots + (q-1)q + q - q^{2n-t-1} + 2 q^{t-1},
\end{multline}
getting a $q$-adic expansion as in Table \ref{QA5}.
\begin{table}[ht]
\centering
\begin{tabular}{|c|c|c|c|c|c|c|c|c|c|c|c|}
  \hline
   & $q^{0}$ & $\cdots$ & $q^{t-2}$ &$q^{t-1}$ & $q^t$ &$\cdots$  & $q^{2n-t-2}$ & $q^{2n-t-1}$ &  $\cdots$ & $q^{2n-2}$ & $q^{2n-1}$ \\ \hline
  $(m-k_1+j_{2,0})_q$ & $0$ &  $\ldots$ & $0$& $2$ & 0& $\ldots$ & $0$ & $q-1$ &$\ldots$ & $q-1$ &  $0$\\
  \hline
\end{tabular}
\caption{$q$-adic expansion of $m-k_1+j_{2,0}$}
\label{QA5}
\end{table}

Therefore the $q$-adic expansion of $m-k_1+j_{2,0}$ has more than two non-vanishing entries and then, it cannot be one of the values $k$ described in Proposition \ref{La5}. Thus $a_{m-k_1+j_{2,0}}=0$. Similarly, the $q$-adic expansion of $j_{2,0}$ has three nonvanishing entries and therefore $a_{j_{2,0}}=0$. This concludes the proof of the case $q >2$ and $s_{i_{2,0}} = s_{m-j_{2,0}} = - s_{i_1} = 1$.

When $q=2$, the only indices $i$ such that $s_i \neq 0$ are $i_0$ and $i_1$. This fact can be proved by noticing that, reasoning as above, the unique  candidate $j_2 : = m -i_2 < m - i_1$ such that $s_{i_2} \neq 0$ is  $j_{2,0} = 1 + q^t = b$. Using again Equality (\ref{threestars}) one gets
\begin{equation*}
 s_{m- b}+  \cdots + a_{k_t} s_{i_1} +  \cdots + a_{i_0+b} s_{i_0} = -(m-b) a_{b}.
\end{equation*}
The only unknown value is $a_{i_0 +b}$, and
\begin{equation*}
i_0 +b= q^{2n-1} + q^{2n-1-t} -(q^{2n-t} +1) + q^{t} +1
= q^{2n-1} - q^{2n-t}+q^{2n-t-1} + q^{t}.
\end{equation*}
Now writing again $q^{2n-1} = \sum_{i=1}^{2n-2} (q-1) q^{2n-1-i} + q$, one deduces that the value $i_0+ b$ equals
\begin{multline*}
(q-1)q^{2n-2}+ \cdots + (q-1) q^{2n-t+1}+ (q-2) q^{2n-t}\\
+ [(q-1) q^{2n-t-1} + \cdots + (q-1) q + q] + q^{2n-t-1} + q^t.
\end{multline*}
The value given in square brackets is $q^{2n-t}$ and therefore the $q$-adic expansion of $i_0 +b$ has more than two non-vanishing entries, which means that $a_{i_0+b}=0$. Thus $s_{i_1} = 1$ and $ -(m-b) a_b =1$ proving that $s_{m-b} =0$. Notice that in this case the characteristic of the ground field is two.
\end{proof}

{\it Step A.4.} To finish the proof of our Theorem \ref{El7} when $t \neq 1$, it suffices to reason as before. That is to say, in the next step define $i_{2,1}:= \min \{ i  \; | \;  i_{2,0} < i \leq m  \mbox{ and $s_i \neq 0$}\}$; again one gets an equality similar  to Equality (\ref{threestars}) and, as we will see later, $a_{m+i_{2,0}-i_{2,1}} \neq 0$  is the only feasible  possibility, then $m+i_{2,0}-i_{2,1} = k_t$ and thus,
$$i_{2,1}= m-(k_t-m+j_{2,0})= m - j_{2,1}= m- (1 + 3 q^{t-1} + (q-3)q^{2n-t-1}).$$
Iterating the reasoning, one obtains candidates $i_{2,\ell}$, $0 \leq \ell \leq q-2$, as in the statement (note that these indices satisfy $i_{2,l}\leq m$, which is equivalent to $j_{2,l}\geq 0$).

We start by computing the values $s_{i_2,\ell}$, $1 \leq \ell \leq q-3$. Recall that $s_{i_2,0}=1$ and note that we assume $q\geq 3$ since we have obtained all the indices $i$ with $s_i\neq 0$ (and the values $s_i$) in the case $q=2$. By Lemma \ref{newton}, the following equality holds:
\begin{multline}
\label{starA}
a_m s_{i_{2,\ell}}+ a_{m+i_{2,\ell-1}-i_{2,\ell}} s_{i_{2,\ell-1}}+ \cdots  \\ + a_{m+i_{2,0}-i_{2,\ell}} s_{i_{2,0}} + a_{m+i_{1}-i_{2,\ell}} s_{i_{1}}+
a_{m+i_{0}-i_{2,\ell}} s_{i_{0}}
= -i_{2,\ell} a_{m- i_{2,\ell}}.
\end{multline}
Consider an index $0 \leq \ell' < \ell$. Then, one gets the chain of equalities
\[
m+ i_{2,\ell'}- i_{2,\ell} = m +(\ell- \ell') q^{t-1} - (\ell- \ell') q^{2n-t-1} = (\ell- \ell') q^{t-1} - (\ell- \ell' -1) q^{2n-t-1} + q^{2n-1},
\]
which proves that $a_{m+i_{2,\ell-1}-i_{2,\ell}}=a_{k_t}$.

The right hand side of Equality (\ref{starA}) vanishes because $a_{m- i_{2,\ell}} = a_{j_{2, \ell}} =0$, which holds since the $q$-adic expansion of $j_{2,\ell}$ does not coincide with any element in Table \ref{TTableoplus}.

Next we are going to show that, with the exception of the first two summands, every summand in the left hand side of Equality (\ref{starA}) vanishes. {\it In this case $s_{i_{2,\ell}} + s_{i_{2,\ell-1}}=0$ and we obtain the values of the indices $s_i$ as in the statement.}

We start  by proving that $a_{m+i_{2,\ell'}-i_{2,\ell}}=0$ whenever $0 \leq \ell' < \ell$ and $\ell' \neq \ell -1$. In this case, $1 < \ell - \ell' \leq q-3$ and the $q$-adic expansion of $m+i_{2,\ell'}-i_{2,\ell}$ is
\[
(\ell - \ell') q^{t-1} - (\ell - \ell' -1) q^{2n-t-1} + q^{2n-1},
\]
which has an expansion with a summand $(q-1)q^{2n-2}$. This implies that $m+i_{2,\ell'}-i_{2,\ell}$ is not an element in Table \ref{TTableoplus} and thus $a_{m+i_{2,\ell'}-i_{2,\ell}}=0$.

We prove now that $a_{m+i_{1}-i_{2,\ell}}=0$. Notice that $m+i_{1}-i_{2,\ell} = i_1+ j_{2,\ell}$. Table \ref{QA3} shows the $q$-adic expansion of $i_1$ and then, the summand corresponding to the least power of $q$ in the $q$-adic expansion of $i_1+ j_{2,\ell}$ is $(l+1)q^{t-1}$. Since $1 \leq \ell \leq q-3$, this last $q$-adic expansion does not coincide with any expansion in Table \ref{TTableoplus}, proving that $a_{m+i_{1}-i_{2,\ell}}=0$.

To conclude the proof of the computation of $s_{i_{2,\ell}}$, $1 \leq \ell \leq q-3$, it only remains to check whether  $a_{m+i_{0}-i_{2,\ell}}$ vanishes. Indeed,
\begin{multline}
m+i_{0}-i_{2,\ell} = i_0 + j_{2,\ell} = i_0 + 1 + (2+ \ell)q^{t-1} + (q-2-\ell) q^{2n-t-1}\\ = q^{2n-1}+q^{2n-t-1}-1-q^{2n-t} + 1 + (2+ \ell)q^{t-1} + (q-2-\ell) q^{2n-t-1}\\ = q + (q-1) q + \cdots + (q-1)q^{2n-2} + (2+ \ell)q^{t-1} +(q- \ell -1) q^{2n-t-1} -q^{2n-t}\\ = (2+\ell) q^{t-1} +(q-\ell -1) q^{2n-t-1} + (q-1)q^{2n-t} + \cdots +(q-1)q^{2n-2}.
\end{multline}
Then, the first summand in the $q$-adic expansion of $i_0 + j_{2,\ell}$ is $(2+\ell) q^{t-1}$ showing that $a_{m+i_{0}-i_{2,\ell}}=0$ by Table \ref{TTableoplus}.

To finish, we determine the value $s_{i_{2,q-2}}$. Then, again by Lemma \ref{newton}, one has
\begin{equation}
\label{delta}
a_m s_{i_{2,q-2}}+ a_{m+i_{2,q-3}-i_{2,q-2}} s_{i_{2,q-3}}+ \cdots +a_{m+i_{0}-i_{2,q-2}} s_{i_{0}}
= -i_{2,q-2} a_{m- i_{2,q-2}},
\end{equation}
where $m-i_{2,q-2}=j_{2,q-2} =b$. Then $-i_{2,q-2} a_{m- i_{2,q-2}} =1$ and, as we proved before, only the first two summands in the left hand side of Equality (\ref{delta}) do not vanish. Then when the characteristic of the ground field is odd, one gets $s_{i_{2,q-2}} +1 = 1$ and thus $s_{i_{2,q-2}} =0$. Otherwise (the characteristic of the ground field is $2$), $s_{i_{2,q-2}} -1 = 1$ and, as well, $s_{i_{2,q-2}} =0$.

{\it Then, we have proved that $q-3$ is the largest index $\ell$ such that $s_{i_{2, \ell}} \neq 0$.}\\

{\it It remains to prove that that $s_i=0$ whenever $m \geq i \geq i_{2,0}$ and $i \neq i_{2,\ell}$, $0 \leq \ell \leq q-3$.}

We start by proving that for any $\ell$ as above, $s_{i_{2,\ell} + j} = 0$ whenever $0 < j < q^{2n-t-1} - q^{t-1}$. Set $f:= i_{2,\ell} - i_{2,\ell-1} = q^{2n-t-1} - q^{t-1}$. By Lemma \ref{newton}, it holds the following equality:
\begin{multline}
\label{noil}
a_m s_{i_{2,\ell}+j}+ a_{m-j} s_{i_{2,\ell}} + a_{m-j-f} s_{i_{2,\ell-1}}+ a_{m-j-2f} s_{i_{2,\ell-2}} + \cdots \\
+ a_{m-j-\ell f} s_{i_{2,0}} + a_{m + i_1 - i_{2,\ell} -j} s_{i_{1}} +
a_{m+i_{0}-i_{2,\ell}-j} s_{i_{0}}
= -(i_{2,\ell} +j) a_{m- i_{2,\ell}- j}.
\end{multline}

Consider an integer $\alpha$ such that $0 \leq \alpha \leq \ell \leq q-3$, then
\[
m-j - \alpha f = q^{2n-t-1} + q^{2n-1} - \alpha (q^{2n-t-1} - q^{t-1}) - j = q^{2n-1} + (1 - \alpha) q^{2n-t-1} + \alpha q^{t-1} -j.
\]
Since $0 < j < q^{2n-t-1} - q^{t-1}$, one gets
\[
q^{2n-1} - \alpha q^{2n-t-1} + (\alpha+1)q^{t-1}< m-j-\alpha f < q^{2n-1} + (1-\alpha)q^{2n-t-1} + \alpha q^{t-1},
\]
and the $q$-adic expansion of $m-j - \alpha f$ contains either the summand $(q-1) q^{2n-2}$ or $q^{2n-1}$. Then, it can be neither $m$ nor $k_t$. This proves that, for all $\alpha$ and $j$ as before, $a_{m-j-\alpha f} = 0$ by Table \ref{TTableoplus}.

Next we prove that  $a_{m + i_1 - i_{2,\ell} -j}$ vanishes. Indeed,
\begin{multline*}
m + i_1 - i_{2,\ell} -j = m + 2m -k_1 -k_t -(m - j_{2, \ell})-j = 2m + j_{2, \ell}-k_1-k_t-j \\ =
2(q^{2n-1} + q^{2n-1-t}) + (1 + (2+ \ell)q^{t-1} + (q-2- \ell)q^{2n-t-1}) - (1 + q^{2n-t}) - (q^{t-1} + q^{2n-1}) - j \\ =
q^{2n-1} - q^{2n-t} + (q- \ell) q^{2n-t-1} + (1+ \ell)q^{t-1}-j.
\end{multline*}
Since $0 < j < q^{2n-t-1} - q^{t-1}$, one gets that
\[
q^{2n-1} - q^{2n-t} + (q- \ell-1) q^{2n-t-1} + (2+ \ell)q^{t-1}
\]
and
\[
q^{2n-1} - q^{2n-t} + (q- \ell) q^{2n-t-1} + (1+ \ell)q^{t-1}
\]
are a lower and an upper bound on the values $m + i_1 - i_{2,\ell} -j$. Then, the $q$-adic expansion of $m + i_1 - i_{2,\ell} -j$ has a summand  $(q-1)q^{2n-2}$ if $\ell \neq 0$. When $\ell=0$, the above mentioned $q$-adic expansion has a summand $q^{2n-1}$ and it must be lower than $k_t$. In any case, $a_{m + i_1 - i_{2,\ell} -j}=0$ for all  $j$.

The value $a_{m + i_0 - i_{2,\ell} -j}$ is also zero.  In fact,
\begin{multline*}
m + i_0 - i_{2,\ell} -j = m +(m-k_1)-(m-j_{2,\ell}) -j =m+j_{2,\ell}-k_1-j \\ =
q^{2n-1}-q^{2n-t} + (q- \ell -1) q^{2n-1-t} + (2 + \ell) q^{t-1} -j.
\end{multline*}
Using again that $0 < j < q^{2n-t-1} - q^{t-1}$, we see that $$q^{2n-1}-q^{2n-t} + (q- \ell -1) q^{2n-1-t} + (2 + \ell) q^{t-1}$$ and $q^{2n-1}-q^{2n-t} + (q- \ell -2) q^{2n-1-t} + (3 + \ell) q^{t-1}$ are an upper and a lower bound on the values $m + i_0 - i_{2,\ell} -j$. Then, computing the $q$-adic expansions of both bounds, we deduce that the $q$-adic expansion of $m + i_0 - i_{2,\ell} -j$ has a summand $(q-1)q^{2n-2}$ and thus $a_{m + i_0 - i_{2,\ell} -j}=0$. Moreover, $-(i_{2,\ell}+j)a_{m- i_{2,\ell}- j}=0$, because when $m- i_{2,\ell}- j$ is not a multiple of $q$, $m- i_{2,\ell}- j$ can be neither $b$ nor $k_1$. Then by (\ref{noil}), it holds that $s_{i_{2,\ell} + j} = 0$ whenever $0 < j < q^{2n-t-1} - q^{t-1}$ and $0 \leq \ell \leq q-3$.

In an analogous manner, it can be shown that $s_{i_{2,q-2} + j} = 0$ for $0 < j \leq 1 +q^t$ and \textit{Theorem \ref{El7} is proved when $t \neq 1$}.\\

{\it Case B. $t = 1$.} In this case, $t=1$, $k_t=k_1$ and an ordered set of indices $i$ candidates to satisfy $s_i \neq 0$ is
\[
m-k_1 < 2(m-k_1) < \cdots < q (m- k_1).
\]
Notice that these are the indices given in the statement because $i_0=m-k_1$, $i_1 = (m-k_1)+(m-k_1) = 2(m-k_1)$ and for $0 \leq \ell \leq q-3$,
\[
(\ell+3) (m-k_1) = m - j_{2,\ell}.
\]
Indeed, $(\ell+3)m - (\ell+3)k_1 = m - j_{2,\ell}$ if and only if $(\ell+3)k_1 =(\ell+2)m + j_{2,\ell}$ if and only if
\[
(\ell+3)\left(1 + q^{2n-1}\right) = (\ell+2) \left(q^{2n-2} + q^{2n-1}\right) + (\ell+3) + \left(q-(2+ \ell)\right)q^{2n-2}.
\]

With this new notation, one has to successively apply Lemma \ref{newton} obtaining equalities as follows for $1 \leq \beta \leq q$:
\begin{multline}
\label{starC}
a_m s_{\beta(m-k_1)}+ a_{k_1} s_{(\beta-1)(m-k_1)}+ a_{2k_1-m} s_{(\beta-2)(m-k_1)} + \cdots \\ + a_{(\beta-1)k_1-(\beta-2)m} s_{m-k_1}
= -\beta (m-k_1) a_{m- \beta(m-k_1)}.
\end{multline}
The following lemma will be useful to conclude our proof.
\begin{lem}
\label{dD}
For $ 1 \leq \alpha \leq q$, the values $\alpha k_1 - (\alpha -1) m$ are equal to $q^{2n-2} (q- \alpha +1) + \alpha$. In addition,  the $q$-adic expansion of $\alpha k_1 - (\alpha -1) m$ coincides with that of some element in Table \ref{TTableoplus} if and only if either $\alpha=1$ or $n=2$ and $\alpha=q$.
\end{lem}
\begin{proof}
The proof follows from the sequence of equalities:
\begin{multline}
\alpha k_1 -(\alpha-1)m= \alpha (1 + q^{2n-1})- (\alpha -1)(q^{2n-2}+ q^{2n-1}) = \\
q^{2n-1} - (\alpha -1) q^{2n-2} + \alpha = q^{2n-2}(q- \alpha +1) + \alpha.
\end{multline}
Moreover, $q^{2n-2}(q- \alpha +1) + \alpha = k_1$ when $\alpha =1$ and it equals $q^2+q$ whenever $n=2$ and $\alpha = q$. Otherwise, the value $q^{2n-2}(q- \alpha +1) + \alpha $ cannot be expressed as $q^i+q^{i+1}$, $0 \leq i \leq 2n-2$.
\end{proof}

{\it Step B.1.} First we assume that
\begin{equation}
\label{supo*}
s_{\beta (m-k_1) + j} = 0\; \mbox{for $0 \leq \beta \leq q-1$ and $0 < j < m-k_1$}.
\end{equation}
In a few lines we will prove that this assumption is true.

Now, from Lemma \ref{dD}, we deduce that every summand (with the exception of the first two ones) in the left hand side of Equality (\ref{starC}) vanishes. In addition, $m - \beta (m -k_1) = m - \beta m + \beta k_1= \beta k_1 - (\beta-1) m$. Again by Lemma \ref{dD}, the right hand side of Equality (\ref{starC}) is zero for $\beta\neq 1$ and $1$ for $\beta=1$. This proves that, {\it whenever $1 \leq \beta \leq q$, the value $s_{\beta(m-k_1)}=1$ if $\beta$ is odd and it equals $-1$, otherwise.}\\

{\it Step B.2.} {\it Let us prove our assertion (\ref{supo*}). Also that $s_{q(m-k_1)+j}=0$ for every $j$ such that $0 < j < j_1:=q^{2n-2}-q^{2n-3}+q-1$ and $s_{q(m-k_1)+j_1}\neq 0$.} Note that $j_1\leq m-k_1\leq m-q(m-k_1)$ and that the first inequality is an equality when $n=2$.

Let us prove Assertion (\ref{supo*}). Take $0 < j < m-k_1=q^{2n-2} -1$ and $0\leq \beta \leq q-1$. By Lemma \ref{newton}, one has
\begin{multline}
\label{Omega}
a_m s_{\beta (m-k_1)+j}+ a_{m-j} s_{\beta (m-k_1)} + a_{k_1-j} s_{(\beta-1) (m-k_1)} + a_{2k_1-m-j} s_{(\beta-2) (m-k_1)} +\cdots \\
+ a_{(\beta-1)k_1-(\beta-2)m-j} s_{m-k_1}
= -(\beta (m-k_1) +j) a_{\beta k_1-(\beta-1)m-j}.
\end{multline}
The values $a_{\alpha k_1 - (\alpha-1)m -j}$, $0\leq \alpha \leq q-1$, satisfy:

\[
a_{\alpha k_1 - (\alpha-1)m -j} = \begin{cases}
0 & \text{ if } 0\leq \alpha\leq q-2 \text{ or } \alpha=q-1 \text{ and } j\neq j_1\\
1 & \text{ if } \alpha=q-1 \text{, } j=j_1 \text{ and } n\neq 2.
\end{cases}
\]
Indeed, for $\alpha=0$ the value $a_{m-j}$ vanishes because $k_1<m-j<m$ and $m$ and $k_1$ are the largest indices $i$ such that $a_i \neq 0$. For $1 \leq \alpha \leq q-1$, from the fact $0 < j < q^{2n-2} -1$ and by Lemma \ref{dD}, it holds the following chain of inequalities:
\[
(q-\alpha) q^{2n-2} + \alpha +1 < \alpha k_1-(\alpha-1)m - j < (q-\alpha+1) q^{2n-2} + \alpha.
\]
Then, looking at the $q$-adic expansions of the bounds on $\alpha k_1-(\alpha-1)m-j$ given by the above inequalities, the result is true for every $\alpha$ such that $1 \leq \alpha <q-1$. When $\alpha=q-1$, $\alpha k_1 - (\alpha-1)m -j$ appears in Table \ref{TTableoplus} only when $n\neq 2$ and $\alpha k_1 - (\alpha-1)m -j=q^{2n-2}+q^{2n-3}$, that is, $j=j_1$.

Hence, from Equality (\ref{Omega}) we have $s_{\beta (m-k_1) + j} = 0$ for $0 \leq \beta \leq q-1$ and $0 < j < m-k_1$ except when $\beta=q-1$, $j=j_1$ and $n\neq 2$. In such a case, Equality (\ref{Omega}) is
\[
s_{(q-1)(m-k_1)+j_1} = -((q-1)(m-k_1)+j_1) a_{(q-1)k_1-(q-2)m-j_1},
\]
and the right hand side vanishes because of the characteristic of the field. This proves Assertion (\ref{supo*}).

Now, when $0<j<j_1$, it holds
\[
q^{2n-3}+1 < qk_1-(q-1)m - j < q^{2n-2}+q,
\]
and again $q k_1 - (q-1)m -j$ appears in Table \ref{TTableoplus} only when $n\neq 2$ and $q k_1 - (q-1)m -j=q^{2n-3}+q^{2n-4}$, that is, $j=j_1-q^{2n-4}+1$. Reasoning as before, we get $s_{q(m-k_1)+j}=0$ for all $0<j<j_1$.

Finally, when $j=j_1$, then $qk_1-(q-1)m-j_1=q^{2n-3}+1$ and
\[
a_{q k_1 - (q-1)m -j_1} = \begin{cases}
0 & \text{ if } n\neq 2\\
1 & \text{ if } n = 2.
\end{cases}
\]
If $n\neq 2$, then $j_1<m-k_1$ and we can apply the results above. From Equality (\ref{Omega}) we get
\[
s_{q(m-k_1)+j_1} + a_{(q-1)k_1-(q-2)m-j_1}s_{m-k_1} = 0,
\]
giving rise to the equality $s_{q(m-k_1)+j_1}=-1$ (note that when $q$ is even, $1=-1$). Otherwise, when $n=2$, $j_1=m-k_1$ and $q(m-k_1)+j_1 = (q+1) (m-k_1)$. Then, Lemma \ref{dD} proves
\begin{equation}
\label{starE}
s_{(q+1) (m-k_1)} + a_{k_1} s_{q(m-k_1)} + a_{q k_1 - (q-1)m} s_{m-k_1}
= -(m-b)a_b.
\end{equation}
Here, the last summand of the left hand side of Equality (\ref{starE}) equals one and the right hand side is also equal to one. Therefore, in this case,
\[
s_{(q+1)(m-k_1)} = -s_{q(m-k_1)}= \begin{cases}
-1 & \text{ if } q \text{ is odd}\\
1 & \text{ if } q \text{ is even}.
\end{cases}
\]

{\it We have proved that the set $\Gamma = \{ \ell (m-k_1)\}_{\ell=1}^q \cup \{q(m-k_1)+j_1\}$ is the set of indices $r \leq q(m-k_1) + j_1$ such that $s_r\neq 0$.}\\

{\it Step B.3.} {\it To conclude we show that if $q(m-k_1) + j_1< r \leq m$, then $s_r=0$ except when $q=n=2$, in which case $s_m\neq 0$.}

Suppose that $q=n=2$ does not hold. In this case, $r=q(m-k_1) + j$ where
\begin{equation}
\label{redo*}
q^{2n-2}-q^{2n-3}+q-1=j_1 <j  \leq m-q(m-k_1)=q^{2n-2} +q,
\end{equation}
and by Lemma \ref{newton}, the following equality
\begin{multline}
\label{cuadra}
a_m s_{q (m-k_1)+j}+ a_{m-j+j_1} s_{q(m-k_1)+j_1}+ a_{m-j} s_{q (m-k_1)} + a_{k_1-j} s_{(q-1) (m-k_1)} + \\ a_{2k_1-m-j} s_{(q-2) (m-k_1)} +\cdots
+ a_{(q-1)k_1-(q-2)m-j} s_{m-k_1}
= -(q (m-k_1) +j) a_{m-q(m-k_1)-j}
\end{multline}
holds.

From the inequalities in (\ref{redo*}), we deduce
\[
q^{2n-1}+q^{2n-3}-q+1 > m-j \geq q^{2n-1}-q,
\]
and then, $a_{m-j}\neq 0$ if and only if $m-j=k_1$. Moreover,
\[
q^{2n-1}+q^{2n-2} > m-j+j_1 \geq q^{2n-1}+q^{2n-2} - (q^{2n-3}+1)>k_1,
\]
which proves $a_{m-j +j_1}=0$.

Recalling Lemma \ref{dD}, for $1 \leq \alpha \leq q$, it holds the chain of inequalities:
\[
(q- \alpha) q^{2n-2} + q^{2n-3} -q +(\alpha+1) > \alpha k_1 - (\alpha-1)m -j \geq (q- \alpha) q^{2n-2}  -q +\alpha.
\]
Looking at the $q$-adic expansions of the bounds on $\alpha k_1 - (\alpha-1)m -j$ given by the above inequalities, one deduces that, when $1\leq \alpha<q$, the coefficient of $q^{2n-2}$ in the $q$-adic expansion of $\alpha k_1 - (\alpha-1)m -j$ admits three possibilities: it is different from $0$ and $1$, it is $0$ in which case the coefficient $q-1$ appears in the before mentioned $q$-adic expansion, or it is $1$ but its contiguous term in the $q$-adic expansion is not $1$. This proves that $a_{\alpha k_1 - (\alpha-1)m -j}=0$ for all $1\leq \alpha<q$.

If $m-j\neq k_1$, then the left hand side of Equality (\ref{cuadra}) is equal to $s_{q(m-k_1) m +j}$. If the right hand side of (\ref{cuadra}) is not a multiple of $q$, then neither $m-q(m-k_1)-j$ is and one could get $a_{m-q(m-k_1)-j}\neq 0$ only when $m-q(m-k_1)-j$ equals $k_1$ or $b$. The first situation cannot hold because $m-q(m-k_1)-j<q^{2n-3}+1<k_1$ and the second one contradicts the fact $m-j\neq k_1$. Therefore, $s_{q(m-k_1) m +j}$ vanishes for all index $j$ as above.

Otherwise, if $m-j=k_1$, one gets
\[
s_{(q+1) (m-k_1)} + a_{k_1} s_{q(m-k_1)} = -(m-b)a_b
\]
and then, $s_{(q+1) (m-k_1)}=0$.\\


The only remaining case is $q=n=2$. Then, $m=12$, $k_1=9$, $j_1=3$ and $\Gamma = \{ \ell (m-k_1)\}_{\ell=1}^{q+1}=\{3,6,9\}$ is the set of indices $r \leq q(m-k_1) + j_1=9$ such that $s_r\neq 0$. Now, Equality (\ref{cuadra}) is
\[
s_{6+j}+a_{15-j}s_9+a_{12-j}s_6+a_{9-j}s_3=-(6+j)a_{6-j},
\]
and, then, one deduces the equalities $s_{10}=s_{11}=0$ and $s_{12}=1$.

This concludes the proof of this case $t=1$ and that of Theorem \ref{El7}.
\end{proof}

With the above ingredients, we are ready to provide the parameters of Hermitian self-orthogonal codes of the type  $E_{\Delta, \mathrm{Tr}}$.

\begin{teo}
\label{era2}
Let $q$ be a prime power. Consider the polynomial $\mathrm{Tr}(X) =\mathrm{Tr}_b(X)$, where $b =1 + q^t$ with $0< t \leq n$, $n$ being a positive integer. Assume that $(q,n,b) \neq (2,2,3)$ is a triple satisfying Property (\ref{A}). Define $A(q,t)$ as follows:
\[
A(q,t):=
\begin{cases}
q^n - \lceil \frac{q-1}{2} \rceil q^{n-1} - \lceil \frac{q-1}{2} \rceil q^{n-t-1} -2 & \text{ if } 0 <t \leq \frac{n}{2}, \\
q^n - \lceil \frac{q-1}{2} \rceil q^{n-1} - \lceil \frac{q-1}{2} \rceil q^{t-1} -2 & \text{ if } \frac{n}{2} <t < n, \\
q^{n-1}-2 & \text{ if } t=n,
\end{cases}
\]
whenever $q \neq 2$, and
\[
A(q,t):=
\begin{cases}
2^n - 2^{t-1} -2 & \text{ when } q=2 \text{ and } t <n,\\
2^{n-1} -2 & \text{ when } q=2 \text{ and } t =n.
\end{cases}
\]
For any nonnegative integer $\tau$, define $\Delta(\tau):=\{a \in \mathbb{Z}\; |\; 0 \leq a \leq \tau\}$. Then, if $\tau \leq A(q,t)$, the linear code $E_{\Delta(\tau), \mathrm{Tr}}$, over $\mathbb{F}_{q^{2n}}$, has length $m = q^{2n-t-1} + q^{2n-1}$ and satisfies:
  \begin{description}
    \item[i)] It is Hermitian self-orthogonal, that is
    \[
    E_{\Delta(\tau), \mathrm{Tr}} \subseteq \left( E_{\Delta(\tau), \mathrm{Tr}} \right)^{\perp_h}.
    \]
    \item[ii)] Its dimension is $\tau +1$ and the minimum distance of $\left( E_{\Delta(\tau), \mathrm{Tr}} \right)^{\perp_h}$ is larger than or equal to $\tau +2$.
  \end{description}
\end{teo}
\begin{proof}
\textit{We first carry out the proof in the case $t \neq n$}.

By Proposition \ref{grado}, $m = q^{2n-t-1} + q^{2n-1}$ is the length of the code $E_{\Delta(\tau), \mathrm{Tr}}$. With the above notation, to show Item i), we have to prove that $\mathrm{ev}_{\mathrm{Tr}}(X^a) \cdot_{h} \mathrm{ev}_{\mathrm{Tr}}(X^b) = 0$ whenever $a, b \leq A(q,t)$. Now
$$
\mathrm{ev}_{\mathrm{Tr}}(X^a) \cdot_{h} \mathrm{ev}_{\mathrm{Tr}}(X^b) = \sum_{j=1}^m \beta_j^{a+q^nb} = s_{a+q^nb}.
$$
Define $\mathcal{I}:= \{i \; | \; 1 \leq i < q^{2n} \mbox{ and $s_i \neq 0$}\}$. With the notation as in Theorem \ref{El7}, set $$\mathcal{I}_1=\{i_0,i_1\} \cup \{i_{2,0},i_{2,1}, \ldots, i_{2,q^*} \},$$ where $q^* =q-3$ if $t \neq 1$ and $q^*=q-2$ otherwise. Theorem \ref{El7} shows that $$\mathcal{I}_1 = \{i \; | \; 1 \leq i \leq m  \mbox{ and $s_i \neq 0$} \}$$ and, clearly, $\mathcal{I}_1 \subseteq \mathcal{I}$.

To prove Item i), we are going to determine a set $\mathcal{J} \supseteq \mathcal{I}$ of candidates $i$ to satisfy $s_i \neq 0$. Since we have obtained the indices $r \leq m$ such that $s_r \neq 0$, we look for indices $r >m$ with this last property. Applying the second part of Lemma \ref{newton}, we get
\[
a_m s_r + a_{m-1} s_{r-1} + \cdots + a_0 s_{r-m} = 0,
\]
so the indices $r>m$ such that $s_r\neq 0$ fulfill $r-m=\beta-\alpha$ for some $\alpha$ and $\beta$ such that $a_\alpha\neq 0$ and $s_\beta\neq 0$. Otherwise, $s_r=0$ by the above formula. Therefore, for obtaining our set of candidates $\mathcal{J}$, we have to consider the elements in $\mathcal{I}_1$ and append those indices $r$ iteratively obtained by the formula $\beta - \alpha + m$, where $\beta$ is some previously obtained element in $\mathcal{J}$ (i.e., $\beta$ is in $\mathcal{I}_1$ or it is a new candidate given by the procedure we are describing) and $\alpha$ is one of the indices (different from $m$)  appearing in Table \ref{TTableoplus}.

The $q$-adic expansion of the indices in $\mathcal{I}_1$ (with the exception of the last one when $t=1$) can be obtained as follows: $i_0 = (q^{2n-1} -1)- (q-1)q^{2n-1-t}$ and then, all the coefficients in its $q$-adic expansion are $q-1$ with the exception of those of $q^{2n-1}$ and $q^{2n-1-t}$, which are zero. The remaining values $i_1, i_{2,0}, \ldots, i_{2,q-3}$ are obtained from the previous one  by adding $q^{2n-1-t} - q^{t-1}$. Thus the coefficients in the $q$-adic expansion  of the elements in $\mathcal{I}_1$ successively decrease one unit with respect to $q^{t-1}$ and increase one unit with respect to $q^{2n-1-t}$. See the forthcoming Table \ref{longtable}, where we give the $q$-adic expansions of the above indices and of an index $j_0$ which will be used later. Each $q$-adic expansion starts in the first part of the table and continues in the corresponding line of the second one.

Now we consider the $q^n$-adic expansion of each index $i \in \mathcal{J}$. It is expressed as  $$i=i(0) + i(1) q^n,$$ with $i(0)$ and $i(1)$ nonnegative integers lower than $q^n$.

We only need  to prove that
\begin{equation}
\label{estrellita}
\text{ if } a,  b \leq A(q,t), \text{ then } a + b q^n \not \in \mathcal{J}.
\end{equation}

For the following reasoning, see Table \ref{longtable}. {\it Assume $t \neq n$ and $q\neq 2$}.

\begin{table}[h]
\begin{tabular}{|c|c|c|c|c|c|c|c|c|c|c|c !{\color{red}\vrule} !{\color{red}\vrule} c|}
  \hline
   & $q^{0}$ & $q^{1}$ & $\cdots$ & $q^{t-1}$ & $q^{t}$ & $\cdots$ & $q^{n-t-1}$ &$q^{n-t}$  & $\cdots$ & $q^{n-2}$ & $q^{n-1}$ & $\color{red}{\rightarrow}$\\ \hline
  $i_0$ & $q-1$ & $q-1$ & $\ldots$ & $q-1$ & $q-1$ & $\ldots$ & $q-1$ & $q-1$ & $\ldots$ & $q-1$ & $q-1$ &\\
  \hline
  $i_1$ & $q-1$ & $q-1$ & $\ldots$ & $q-2$ & $q-1$ & $\ldots$ & $q-1$ & $q-1$ & $\ldots$ & $q-1$ & $q-1$ & \\\hline
  $\vdots$ & $\vdots$ & $\vdots$ & $\ldots$ &  $\vdots$ &  $\vdots$ & $\ldots$ & $\vdots$& $\vdots$&  $\ldots$ &  $\vdots$ & $\vdots$ & \\ \hline
  $i_{2,q-3}$& $q-1$ & $q-1$ & $\ldots$ & $0$ & $q-1$ & $\ldots$ & $q-1$ & $q-1$ & $\ldots$ & $q-1$ & $q-1$ & \\ \hline
  $j_0$ & $q-1$ & $q-1$ & $\ldots$ & $q-1$ & $q-1$ & $\ldots$ & $\lfloor\frac{q-1}{2}\rfloor$ & $q-1$ & $\ldots$ & $q-1$ & $\lfloor\frac{q-1}{2}\rfloor$ & \\ \hline
\end{tabular}\hspace{0pt}%
\begin{tabularx}{0.983\textwidth}{|>{\centering\arraybackslash}X !{\color{red}\vrule} !{\color{red}\vrule} >{\centering\arraybackslash} X | >{\centering\arraybackslash} X | >{\centering\arraybackslash} X | >{\centering\arraybackslash} X | >{\centering\arraybackslash} X |>{\centering\arraybackslash} X |>{\centering\arraybackslash} X|>{\centering\arraybackslash} X|>{\centering\arraybackslash} X|>{\centering\arraybackslash} X|>{\centering\arraybackslash} X|}
  \hline
   $\color{red}{\rightarrow}$ & $q^{n}$ & $q^{n+1}$ & $\cdots$ & $q^{n+t-1}$ & $q^{n+t}$ & $\cdots$ & $q^{2n-t-1}$ &$q^{2n-t}$  & $\cdots$ & $q^{2n-2}$ & $q^{2n-1}$\\ \hline
   & $q-1$ & $q-1$ & $\ldots$ & $q-1$ & $q-1$ & $\ldots$ & $0$ & $q-1$ & $\ldots$ & $q-1$ & $0$ \\
  \hline
   & $q-1$ & $q-1$ & $\ldots$ & $q-1$ & $q-1$ & $\ldots$ & $1$ & $q-1$ & $\ldots$ & $q-1$ & $0$ \\\hline
   & $\vdots$ & $\vdots$ & $\ldots$ &  $\vdots$ &  $\vdots$ & $\ldots$ & $\vdots$& $\vdots$& $\ldots$ &  $\vdots$ & $\vdots$\\ \hline
   & $q-1$ & $q-1$ & $\ldots$ & $q-1$ & $q-1$ & $\ldots$ & $q-1$ & $q-1$ & $\ldots$ & $q-1$ & $0$\\ \hline
   & $q-1$ & $q-1$ & $\ldots$ & $q-1$ & $q-1$ & $\ldots$ & $\lfloor\frac{q-1}{2}\rfloor$ & $q-1$ & $\ldots$ & $q-1$ & $\lfloor\frac{q-1}{2}\rfloor$\\ \hline
 \end{tabularx}\hspace{0pt}%
\caption{$q$-adic expansions of the candidates in $\mathcal{I}_1$ and $j_0$ when $1<t\leq \frac{n}{2}$}
\label{longtable}
\end{table}

Notice that, as said, $i_0(1), i_1(1)$ and $i_{2, \ell}(1)$, $0 \leq \ell \leq q-3$, are positive integers lower than $q^{n-1}$. 
To obtain the values in $\mathcal{J}$, we have to add $m= q^{2n-1} + q^{2n-1-t}$ to every previous value in $\mathcal{J}$ and subtract an index different from $m$ appearing in Table \ref{TTableoplus}. In addition, in each step, we start with an index $\beta$ with $q^n$-adic expansion $i(0) + i(1) q^n$  to get the $q^n$-adic expansion of the next one: $j(0) + j(1) q^n=\beta-\alpha+m$, where $j(1) > i(1)$ and $j(0) \leq i(0)$.

\textit{Suppose now that $1 \leq t \leq \frac{n}{2}$}. Our first step is to define a bound $A'(q,t)\geq A(q,t)$, that eases the understanding of the definition of $A(q,t)$. We desire that every $i\in\mathcal{J}$ fulfills the following statement:
\begin{equation}
\label{cota}
\text{ if } \min\{i(0),i(1)\}\leq A'(q,t), \text{ then } \max\{i(0),i(1)\}>A'(q,t),
\end{equation}
because then (\ref{estrellita}) is proved for $A'(q,t)$ instead of $A(q,t)$ (and therefore for $A(q,t)$) and thus Theorem \ref{era2} i) holds. The reason for choosing $A(q,t)$ is that, as we will see next, the expression of $A'(q,t)$ differs depending on whether or not
$q-1$ is even.

The bound $A'(q,t)=\sum_{k=0}^{q^{n-1}}A'_kq^k$ is given by $\min\{\max\{i(0),i(1) \mid i\in\mathcal{J}\}-1$. To get this minimum, the elements $i\in\mathcal{J}$ to be considered are of the form:
\begin{equation}\label{bound}
    i=i'+xm-x\alpha,
\end{equation}
$i'\in\mathcal{J}$, $\alpha=q^{n-1}+q^{n-t-1}$ or $\alpha=q^{n-1}+q^{n+t-1}$ and $x<q$ is a positive integer such that, in the expression $|i(0)-i(1)|=\sum_{k=0}^{q^{n-1}}h_kq^k$, $h_{q^{n-1}}$ is either $0$ (if $q-1$ is even) or $1$ (if $q-1$ is odd). When looking for such a bound, one should not consider candidates given by values $i' \in \mathcal{J}\backslash \mathcal{I}_1\cup\{i_{2,q-2}\}$; this is because the coefficient of $q^{2n-1}$ in its $q$-adic expansion is positive and increases if one adds $m$, giving rise to a greater value $\max\{i(0),i(1)\}$ than when starting with elements $i'$ in $\mathcal{I}_1\backslash\{i_{2,q-2}\}$.

Then, to get the bound $A'(q,t)$, one must keep in mind those elements $i\in\mathcal{J}$ described in (\ref{bound}) with $i'\in\mathcal{I}_1\backslash\{i_{2,q-2}\}$ and consider the element $i$ whose value $\max\{i(0),i(1)\}$ is a minimum. Taking into account that the coefficients of $q^{n-1}$ and $q^{2n-1}$ in the $q$-adic expansion of every element in $\mathcal{I}_1\backslash\{i_{2,q-2}\}$ are, respectively, $q-1$ and $0$, we get $A'_{q^{n-1}}=\frac{q-1}{2}$ when $q-1$ is even and $A'_{q^{n-1}}=\lceil\frac{q-1}{2}\rceil$, otherwise.


As a consequence, the inclusion in our Item i) holds if one considers the value $A(q,t):=\min S-1$, where
\begin{multline*}
    S=\Bigg\{j(0) \mid j(0)=j(1),\text{ } j(0)=\sum_{k=0}^{q^{n-1}-1}j(0)_kq^k+\floor*{\frac{q-1}{2}} q^{n-1},\\
    j=i+xm-y(q^{n-1}+q^{n-t-1}) \text{ or } j=i+xm-y(q^{n-1}+q^{n+t-1}),\\ i\in\mathcal{I}_1\backslash\{i_{2,q-2}\},\text{ } x, y \text{ positive integers}\Bigg\}.
\end{multline*}
Then, $A(q,t)=j_0(0)-1$, where
$$
j_0=i_0 +  \left\lfloor \frac{q-1}{2} \right\rfloor m - \left\lceil \frac{q-1}{2} \right\rceil \left(q^{n-t-1} + q^{n-1}\right),
$$
see Table \ref{longtable}, and notice that we need $i=i_0$ to obtain $j(0)=j(1)$. This proves the corresponding case in the statement. It is worthwile to add that $A'(q,t)=A(q,t)$ only when $q-1$ is even.


\textit{When $\frac{n}{2} < t <n$}, then $n-t-1 < t-1$ and we can reason similarly, but to get the value $j_0'$ playing the same role as $j_0$, instead of $i_0$, we have to use the value $i_0'$ in $\mathcal{I}_1$ whose $q$-adic expansion has $\lceil \frac{q-1}{2} \rceil$ as a coefficient for $q^{2n-t-1}$. That is, $i_0'$ is $i_{2, \lceil \frac{q-1}{2} \rceil -2}$ when $q \geq 4$ and $i_0'=i_1$ otherwise ($q=3$). Then we have to consider
\[
j_0' = i_0' +  \left\lfloor \frac{q-1}{2} \right\rfloor m - \left\lceil \frac{q-1}{2} \right\rceil \left(q^{n-1} + q^{n+t-1}\right)
\]
to deduce that $A(q,t) = j_0'(0)-1$.



\textit{The cases $t=n$ and $q=2$} follow by computing $$A(q,t)=\min\{\max\{i(0),i(1)\}, i\in\mathcal{J}\}-1=i_1(0)-1.$$

Now, we show Item ii). The dimension of $E_{\Delta(\tau), \mathrm{Tr}}$ is $\tau +1$ because $\mathrm{ev}_{\mathrm{Tr}}$ is injective; in fact it is given by a Vandermonde matrix over $\mathbb{F}_{q^{2n}}$. Finally, the assertion about the minimum distance of $\left( E_{\Delta(\tau), \mathrm{Tr}} \right)^{\perp_h}$ can be proved as in the proof of \cite[Theorem 8]{Traza}.
\end{proof}

A well-known result (see for instance \cite{Aly, Ketkar}) states that a self-orthogonal (for the Hermitian inner product) linear code $\mathcal{C}$ over a field $\mathbb{F}_{Q^2}$ of length $n$ and dimension $k$, $Q$ being a prime power, gives rise to a stabilizer quantum code with parameters $[[n, n-2k, \geq d_h^{\perp}]]_Q$, where $d_h^{\perp}$ stands for the minimum distance of the Hermitian dual of $\mathcal{C}$. Thus, we get the following immediate consequence of Theorem \ref{era2}.

\begin{cor}
\label{el15}
Let $q$ be a prime power. Assume that $(q,n,b)$ is a triple satisfying Property (\ref{A}). With notation as in Theorem \ref{era2}, for each non-negative integer $\tau \leq A(q,t)$, there is a stabilizer quantum code with parameters
\[
[[q^{2n-1-t}+q^{2n-1}, q^{2n-1-t}+q^{2n-1}-2 \tau-2, \geq \tau +2]]_{q^n}.
\]
\end{cor}

\section{Stabilizer quantum codes from subfield-subcodes. Examples}

In this section we show that, considering subfield-subcodes of the above described codes, one  obtains $q^{n'}$-ary stabilizer codes, where $n' <n$ and $n'$ divides $n$. We also prove that some of these codes have excellent parameters as we will explain at the end of the section.

\subsection{Subfield-subcodes}
\label{subfield}
Denote by $(\alpha_1, \alpha_2, \ldots, \alpha_{q^{2n}-1})$ an ordered sequence of the non-zero elements of the finite field $\mathbb{F}_{q^{2n}}$. Similarly to (\ref{evaluation}), we consider the map
\begin{equation}
\label{evaluation2}
\mathrm{ev}: R \left(:= \mathbb{F}_{q^{2n}}[X] /\langle X^{q^{2n}-1} -1  \rangle\right) \to \mathbb{F}_{q^{2n}}^{q^{2n}-1},
\end{equation}
given by $\mathrm{ev}(h)=(h(\alpha_1), \ldots, h(\alpha_{q^{2n}-1}))$,
$h$ being the polynomial function defined by the class of a polynomial (also named $h$) of $\mathbb{F}_{q^{2n}}[X]$ in $\mathbb{F}_{q^{2n}}[X]/\langle X^{q^{2n}-1} -1  \rangle$.

Fix a positive integer $n' <n$ such that $n'$ divides $n$. Write $\mathfrak{G}:=\{0,1, \ldots, q^{2n}-2\}$ regarded as a set of representatives of the quotient ring $\mathbb{Z}_{q^{2n}-1}$, where we consider cyclotomic cosets with respect to $q^{2n'}$. Denote by $\mathfrak{C}_g$, $g \in \mathfrak{G}$, the cyclotomic coset of the form $\mathfrak{C}_g:=\{q^{2 n' i}g \mid i \geq 0\}$, where for simplicity we only pick those integers $g$ such that $g= \min \mathfrak{C}_g$. These values $g$ are named minimal representatives of their cyclotomic cosets. Denote by $\mathfrak{S}:=\{g_0<g_1< \cdots <g_{\omega}\}$ the ordered set of minimal representatives of cyclotomic cosets.

Let $\Gamma$ be a nonempty subset of $\mathfrak{G}$, using the map $\mathrm{ev}$, we define the evaluation code $\mathcal{C}_\Gamma$ as the linear code, over the field $\mathbb{F}_{q^{2n}}$, generated by the set $\{ \mathrm{ev}(X^i) \;|\; i \in \Gamma\}$. Within this framework, the subfield-subcode of $\mathcal{C}_\Gamma$ over the field $\mathbb{F}_{q^{2 n'}}$ is
\[
\mathcal{C}_\Gamma|_{\mathbb{F}_{q^{2 n'}}} := \mathcal{C}_\Gamma \cap \left(\mathbb{F}_{q^{2 n'}}\right)^{q^{2n}-1}.
\]

In this article, to bound the minimum distance by considering the BCH approach, we are only interested in sets $\Gamma$ which are union of cyclotomic cosets whose minimal representatives start at $g_0$ and are consecutive. That is, fix a positive integer $\tau < \omega$ and set
\[
\Gamma(\tau):= \mathfrak{C}_{g_0} \cup \mathfrak{C}_{g_1} \cup \cdots \cup \mathfrak{C}_{g_\tau}.
\]
The subfield-subcodes over $\mathbb{F}_{q^{2n'}}$ of the codes $\mathcal{C}_{\Gamma(\tau)}$ enjoy good properties with respect to duality since they are Galois invariant (see \cite{Sti} and \cite[Proposition 11]{Traza}) and, as a consequence, one gets the following result:
\begin{pro}
With the above notation, the following BCH-type bound holds:
\[
d \left[ \left(\mathcal{C}_{\Gamma(\tau)}|_{\mathbb{F}_{q^{2 n'}}}\right)^{\perp_h} \right] \geq g_{\tau +1} +1,
\]
where $d \left[ \left(\mathcal{C}_{\Gamma(\tau)}|_{\mathbb{F}_{q^{2 n'}}}\right)^{\perp_h} \right] $ denotes the minimum distance of the Hermitian dual of the subfield-subcode over the field $\mathbb{F}_{q^{2 n'}}$ of the code $\mathcal{C}_{\Gamma(\tau)}$.
\end{pro}

Now we return to consider the family of evaluation codes defined in Section \ref{se:dos}, where the map $\mathrm{ev}_{\mathrm{Tr}}$ is used. Keeping the notation as in that section, we define the subfield-subcode over the field $\mathbb{F}_{q^{2 n'}}$ of the code $E_{\Delta, \mathrm{Tr}}$ as
\[
E_{\Delta, \mathrm{Tr}}|_{\mathbb{F}_{q^{2 n'}}} := E_{\Delta, \mathrm{Tr}} \cap \mathbb{F}_{q^{2 n'}}^m.
\]

Recall that $\mathcal{H}=\{0,1, \ldots, m-1\}$ and pick $g_\tau $. Consider the set $\Gamma(\tau)$; these values are initially regarded as powers of monomials generating a linear space of elements in the quotient ring $R$, but, since we desire to use the map $\mathrm{ev}_{\mathrm{Tr}}$, we must consider their classes modulo the ideal  $\langle \mathrm{Tr}(X) \rangle$ and, when considered as elements in  $\mathbb{F}_{q^{2n}}[X]/\langle \mathrm{Tr}(X) \rangle$, they provide generators of the form $X^i$, $i  \in \Gamma(\tau)^{\mathcal{H}} $, of a linear space which can be evaluated by $\mathrm{ev}_{\mathrm{Tr}}$. Notice that $\Gamma(\tau)^{\mathcal{H}}$ is a suitable set of indices included in $\mathcal{H}$. Then, reasoning as in the proof of \cite[Theorem 13]{Traza}, the following result follows:
\begin{pro}
\label{subf}
The dimension and minimum distance of the subfield-subcode $E_{\Gamma(\tau)^{\mathcal{H}}, \mathrm{Tr}}|_{\mathbb{F}_{q^{2 n'}}}$ and its Hermitian dual satisfy:
\begin{enumerate}
  \item $\dim\left( E_{\Gamma(\tau)^{\mathcal{H}}, \mathrm{Tr}}|_{\mathbb{F}_{q^{2 n'}}} \right) \leq \sum_{\ell=0}^{\tau} \# \mathfrak{C}_{g_\ell}$, where $\#$ means cardinality.
  \item $d \left( E_{\Gamma(\tau)^{\mathcal{H}}, \mathrm{Tr}}|_{\mathbb{F}_{q^{2 n'}}} \right)^{\perp_h} \geq g_{\tau+1} +1$.
\end{enumerate}
\end{pro}

We conclude this subsection by providing parameters of $q^{n'}$-ary stabilizer quantum codes, derived from the codes in Proposition \ref{subf}.

\begin{teo}
\label{bueno}
Let $q$ a prime power and $(q,n,b) \neq (2,2,3)$, $b= 1 + q^t$ and $0<t \leq n$, a triple satisfying Property (\ref{A}). Fix a positive integer $n' < n$ such that $n'$ divides $n$. Set
$\mathfrak{S}:=\{g_0<g_1< \cdots <g_{\omega}\} \subset \mathfrak{G}:=\{0,1, \ldots, q^{2n}-2\}$ the ordered set of minimal representatives of cyclotomic cosets, corresponding to the quotient ring $\mathbb{Z}_{q^{2n}-1}$, with respect to $q^{2n'}$. Consider the value $A(q,t)$ introduced in Theorem \ref{era2} and the following values:
\[
B(q,t) := q^n - (q-1) q^{n-t} -q, \;\;\; B^1(q,t) := q^n - (q-1) q^{n-t} -2
\]
and  $C(q,t) := (q^{2n-2}-1)/(q^{n-2}+1)$.

Define $D(q,t)$ as follows:
\begin{itemize}
\item When $t >1$,
\begin{itemize}
\item $D(q,t):=A(q,t)$, whenever $n' \neq 1$.
\item Otherwise ($n' =1$):
\[
D(q,t):= \begin{cases} B(q,t) & \text{ if } n \text{ is even}, \\
\min\{A(q,t),B(q,t)\} & \text{otherwise.}
\end{cases}
\]
\end{itemize}
\item When $t =1$ and $n \neq 2$,
\begin{itemize}
\item $D(q,t):=A(q,t)$, whenever $n' >2$,
\item $D(q,t):=C(q,t)$, whenever $n' =2$,
\item $D(q,t):=B^1(q,t)$, otherwise ($n' =1$).
\end{itemize}
\item When $t =1$ and $n = 2$, $D(q,t) := q-2$.
\end{itemize}

Then, for each element $g_\tau \in \mathfrak{S}$ such that  $g_\tau \leq D(q,t)$, the subfield-subcode $E_{\Gamma(\tau)^{\mathcal{H}}, \mathrm{Tr}}|_{\mathbb{F}_{q^{2 n'}}}$ is Hermitian self-orthogonal and, as a consequence, there exists a stabilizer quantum code with parameters
\[
\left[ \left[ q^{2n-1-t} + q^{2n-1}, \geq q^{2n-1-t} + q^{2n-1} - 2 \sum_{\ell=0}^{\tau} \#\mathfrak{C}_{g_\ell}, \geq g_{\tau +1} +1 \right] \right]_{q^{n'}}.
\]
\end{teo}

\begin{proof}
Our proof follows a close reasoning to that used when proving \cite[Theorem 15]{Traza}, although we consider subfield-subcodes over $\mathbb{F}_{q^{2 n'}}$ instead of over $\mathbb{F}_{q^{2}}$. Consider the basis $\mathcal{B}$ of $\mathcal{C}_{\Gamma(\tau)}$ introduced in the proof of \cite[Proposition 11]{Traza} and, reasoning as at the beginning of the proof of \cite[Theorem 15]{Traza}, it suffices to prove that
\[
\mathrm{ev}_{\mathrm{Tr}} \left( X^{a q^{2 n' \ell} + b q^{n'} q^{2n'm}} \right) \cdot \mathrm{ev}_{\mathrm{Tr}} \left( X^{0}\right) = 0,
\]
for values $\ell, m \in \{0,1, \ldots, \frac{n}{n'} -1\}$, $m \geq \ell$ and $a, b \in \{g_0, g_1, \ldots, g_\tau\}$.

When $n' (2m -2 \ell +1) \leq n$, it holds that
\[
\mathrm{ev}_{\mathrm{Tr}} \left( X^{a q^{2 n' \ell} + b q^{n'} q^{2n'm}} \right) \cdot \mathrm{ev}_{\mathrm{Tr}} \left( X^{0}\right) = \left[ \mathrm{ev}_{\mathrm{Tr}} \left( X^{a + b q^{n' (2m-2\ell +1)}} \right) \cdot \mathrm{ev}_{\mathrm{Tr}} \left( X^{0}\right)
\right]^{q^{2n'\ell}}.
\]
Otherwise, $n <n' (2m -2 \ell +1) \leq n' ( \frac{2n}{n'} -1) = 2n- n' < 2n$. Therefore, one can set $2m -2 \ell +1 = \frac{n}{n'} + s$, where $1 \leq s < \frac{n}{n'}$ and then
\begin{multline}
\mathrm{ev}_{\mathrm{Tr}} \left( X^{a q^{2 n' \ell} + b q^{n'} q^{2n'm}} \right) \cdot \mathrm{ev}_{\mathrm{Tr}} \left( X^{0}\right) \\
= \left[ \mathrm{ev}_{\mathrm{Tr}} \left( X^{a + b q^{n +s n'}} \right) \cdot \mathrm{ev}_{\mathrm{Tr}} \left( X^{0}\right)
\right]^{q^{2 n' \ell}}\\
= \left( \left[ \mathrm{ev}_{\mathrm{Tr}} \left( X^{a q^{n - s n'}+b} \right) \cdot \mathrm{ev}_{\mathrm{Tr}} \left( X^{0}\right)
\right]^{q^{n +s n'}}
\right)^{q^{2 n' \ell}}.\\
\end{multline}
Thus, one concludes that it suffices to prove that both products
\[
\mathrm{ev}_{\mathrm{Tr}} \left( X^{a + b q^{n'r }} \right) \cdot \mathrm{ev}_{\mathrm{Tr}} \left( X^{0}\right)
\]
and
\[
\mathrm{ev}_{\mathrm{Tr}} \left( X^{a q^{n' r} + b} \right) \cdot \mathrm{ev}_{\mathrm{Tr}} \left( X^{0}\right)
\]
vanish for all values $a, b \leq D(q,t)$ and $n' r \leq n$. Then, since we give a common bound for $a$ and $b$, it suffices to check that
\begin{equation}
\label{TT}
\mathrm{ev}_{\mathrm{Tr}} \left( X^{a + b q^{n'r }} \right) \cdot \mathrm{ev}_{\mathrm{Tr}} \left( X^{0}\right) = 0
\end{equation}
for $a, b \leq D(q,t)$ and $0 \leq r \leq \frac{n}{n'}$.

{\it Assume first that $t >1$}. Then $a + b q^{n' r} < i_0 = q^{2n-1} - (q-1) q^{2n-1-t} -1$ when $n' r < n-1$ and Equality (\ref{TT}) holds by Theorem \ref{El7}. Thus, one only has to check Equality (\ref{TT}) when $n'r =n$ or $n'r = n-1$. {\it Suppose first that $n' \neq 1$}, then the case $n'r = n-1$ does not happen because $n'$ divides $n$. Then, $n'r=n$ and therefore, if $a, b \leq A(q,t)$, Equality (\ref{TT}) is true because of the proof of Theorem \ref{era2} and our result is proved in this case. {\it Now, if $n'=1$}, it suffices that $a, b \leq A(q,t)$ to prove (\ref{TT})  when $n' r = r= n$. Otherwise, $n'r= r= n-1$, and $a, b \leq B(q,t)$ implies $a + b q^{n-1} < i_0$. Noticing that $(B(q,t) +1) + (B(q,t) +1) q^{n-1} \geq i_0$, one deduces that (\ref{TT})  is true whenever $a, b \leq \min \{A(q,t), B(q,t)\}$. Note that if $n$ is even, $r=n$ means $2m - 2\ell +1 =n$ by the reasoning at the beginning of the proof and this case cannot hold. Hence, when $t >1$, $n'=1$ and $n$ is even, the bound $D(q,t)$ equals $B(q,t)$.

To conclude the proof, {\it assume that $t=1$.} First {\it suppose $n\neq 2$}. Then $i_0 = q^{2n-2} -1$ and $a + b q^{n'r} < i_0$ when $n'r < n-2$. As above, to prove Equality (\ref{TT}) when $n'r=n$ it suffices to have that $a, b \leq A(q,t)$. But one needs $a, b \leq B^1(q,t)$ in case $n'r=n-1$ and $a, b \leq C(q,t)$ whenever $n'r=n-2$ (notice that $C(q,t)(1+q^{n-2})=i_0$ but $\lfloor C(q,t)\rfloor\neq C(q,t)$). We also notice that
\[
B^1(q,t) < C(q,t).
\]
Finally, since $n=n' \alpha $ for some positive integer $\alpha$, $n' r = n-2 = n' \alpha -2$, then $2= n' ( \alpha - r)$, which happens only when either $n'=1$ or $n'=2$. Therefore, $D(q,t)$ equals $A(q,t)$ when $n' >2$, it is $C(q,t)$ when $n'=2$ and $B^1(q,t)$ in case $n'=1$. {\it When $n=2$}, then $n'=1$ and $B^1(q,t)=q-2$. This concludes the proof.
\end{proof}

\begin{rem}\label{remdim}
When the set of exponents of the polynomial $\mathrm{Tr}(X)$ is contained in $\Gamma(\tau)$, then $\dim\left(E_{\Gamma(\tau)^{\mathcal{H}}, \mathrm{Tr}}|_{\mathbb{F}_{q^{2 n'}}}\right) \leq \sum_{\ell=0}^{\tau} \# \mathfrak{C}_{g_\ell}-1$ (because there is a relation modulo $\mathrm{Tr}(X)$ ($\mathrm{Tr}(X)=0$), which decreases the dimension by one). Therefore, one gets a favourable situation since the bound on the dimension of the stabilizer quantum code given in Theorem \ref{bueno} is increased by two.
\end{rem}

\subsection{Examples}
\label{Ejemplos}
In this subsection, we present some examples of quantum stabilizer codes obtained from our previous results. We only show those codes having good parameters, in particular the parameters of all codes in this subsection exceed the quantum Gilbert-Varshamov bounds \cite{FengMa, Ketkar, 71kkk} and, some of them, either are binary records or improve the parameters of others available in the literature.

\subsubsection{Binary examples}
\label{binary}

With the notation as in Theorem \ref{bueno}, consider the triple $(q,n,b)=(2,4,5)$. We have that $t=2$ and  $g_0=0$, $g_1=1$, $g_2=2$, $g_3=3$, $g_4=5$, $g_5=6$, $g_6=7$, $g_7=9$, $g_{8}=10$, $g_{9}=11$ are minimal representatives of cyclotomic cosets. Moreover the cardinality of the cyclotomic cosets $\mathfrak{C}_{g_\ell}$, $1 \leq \ell \leq 8$
is always $4$. Set $n'=1$ and then $D(q,t)=B(q,t) =10$. Applying Theorem \ref{bueno} with $\tau=8$ and noticing that the condition in Remark \ref{remdim} holds, we obtain a $[[160, 96, \geq 12]]_2$ binary stabilizer quantum code, which improves the $[[160,96, \geq 11]]_2$ code given in \cite{codetables}. Thus, we have obtained a record as a binary quantum code. In this paper, by record, we mean a binary quantum code whose parameters either improve some given in \cite{codetables} or correspond to an entry in \cite{codetables} whose construction was missing. Now, using the propagation rules in \cite[Lemmas 69 and 71]{Ketkar} that state that if there exists an $[[n,k,d]]_q$ quantum code, then there are  quantum codes with parameters $[[n,k-1,\geq d]]_q$ and $[[n+1,k, d]]_q$, we find four new records. These are  $[[160,95, \geq 12]]_2$,  $[[161,96, \geq 12]]_2$, $[[162, 96, \geq 12]]_2$ and $[[163, 96, \geq 12]]_2$.

In the remaining of this paper, we will also use the following result to construct stabilizer codes. This result was stated in \cite{galher} and it is an easy consequence of \cite[Lemma 76]{Ketkar} (see also \cite{AK}).

\begin{teo}
\label{Eldearxiv}
Let $C$ be an $\mathbb{F}_{q^{2r}}$-linear code of length $n$ and dimension $k$, where $r$ is a positive integer. Suppose $C \subseteq C^{\perp_h}$, where
\[
C^{\perp_h} := \left\{  \boldsymbol{x} \in \left(\mathbb{F}_{q^{2r}}\right)^n \; \mid \; \boldsymbol{x} \cdot_h \boldsymbol{y} := \sum_{i=1}^n x_i y_i^{q^{r}} = 0 \mbox{ for all $\boldsymbol{y}$ in $C$}
\right\}.
\]
Then, there exists an $\mathbb{F}_{q}$-stabilizer quantum code with parameters
\[
\big[\big[ rn, rn - 2r k, \geq d^\perp_h  \big]\big]_q,
\]
where $d^\perp_h$ is the minimum distance of the code $C^{\perp_h}$.
\end{teo}

With the same previous triple $(q,n,b)=(2,4,5)$, using Theorem \ref{era2} and sets $\Delta(i)$, $0 \leq i \leq 12$ ($A(q,t)=12$), one obtains Hermitian self-orthogonal codes $E_{\Delta(i), \mathrm{Tr}}$. By applying Theorem \ref{Eldearxiv} to these codes, one gets binary quantum error-correcting codes of length $n=640$ whose parameters are displayed in Table \ref{latablabin1}.

\begin{table}[b]
\begin{center}
\begin{tabular}{||c|c|c|c|c|c|c|c|c|c|c|c|c||}
  \hline
  $k$ &  624 & 616 & 608 & 600 & 592 &  584 & 576 & 568 & 560& 552& 544&  536   \\
   \hline
 $d\geq$  & 3 & 4 &  5 & 6 & 7 & 8 & 9 & 10 & 11 & 12& 13& 14 \\
  \hline
\end{tabular}
\end{center}
\caption{Parameters of binary quantum codes of length 640}
\label{latablabin1}
\end{table}

We can also combine our procedures and starting from our initial linear codes over $\mathbb{F}_{2^8}$, we consider subfield-subcodes over $\mathbb{F}_{2^4}$. These codes use sets $\Delta$ which are successive union of consecutive cyclotomic cosets $\mathfrak{C}^{'}_i$, $0 \leq i \leq 12$, with respect to $2^4$. That is $\mathfrak{C}^{'}_0 = \{0\}$, $\mathfrak{C}^{'}_1 = \{1,16\}$, $\mathfrak{C}^{'}_2 = \{2, 32 \}$, $\mathfrak{C}^{'}_3 = \{3, 48 \}$, $\mathfrak{C}^{'}_4 = \{4, 64 \}$, $\mathfrak{C}^{'}_5 = \{5, 80 \}$, $\mathfrak{C}^{'}_6 = \{6, 96 \}$, $\mathfrak{C}^{'}_7 = \{7, 112 \}$, $\mathfrak{C}^{'}_8 = \{8, 128 \}$, $\mathfrak{C}^{'}_9 = \{9, 144 \}$, $\mathfrak{C}^{'}_{10} = \{10, 160 \}$, $\mathfrak{C}^{'}_{11} = \{11, 176 \}$, and $\mathfrak{C}^{'}_{12} = \{12, 192 \}$. In this way, using Theorem \ref{bueno} with $n'=2$, one obtains  Hermitian self-orthogonal codes over $\mathbb{F}_{16}$. Note that $D(q,t)=12$. Then, Theorem \ref{Eldearxiv}, applied to these codes, gives rise to binary quantum error-correcting codes of length $n=320$. Some of their parameters are displayed in Table \ref{latablabin2}.

\begin{table}[b]
\begin{center}
\begin{tabular}{||c|c|c|c|c|c|c|c|c|c|c|c|c||}
  \hline
  $k$ & 308 & 300 & 292 & 284 & 276 & 268 & 260 & 252 & 244 & 236 & 228 & 220  \\
   \hline
 $d\geq$ &  3 & 4 &  5 & 6 & 7 & 8 & 9 & 10 & 11 & 12& 13& 14 \\
  \hline
\end{tabular}
\end{center}
\caption{Parameters of binary quantum codes of length 320}
\label{latablabin2}
\end{table}

Applying the last two procedures (both with $A(q,t)=10$) to the triple $(q,n,b) = (2,4,9)$, we get binary quantum error-correcting codes with parameters
$$[[576, 576 -8(i+1), \geq i+2]]_2$$
with $1 \leq i \leq 10$, and of length $n=288$ with parameters as in Table \ref{latablabin4}.



\begin{table}[b]
\begin{center}
\begin{tabular}{||c|c|c|c|c|c|c|c|c|c|c||}
  \hline
  $k$ & 276& 268 & 260 & 252 & 244 & 236 & 228 & 220 & 212 & 204  \\
   \hline
 $d\geq$ &  3 & 4 &  5 & 6 & 7 & 8 & 9 & 10 & 11 & 12 \\
  \hline
\end{tabular}
\end{center}
\caption{Parameters of binary quantum codes of length 288}
\label{latablabin4}
\end{table}

The lengths of the last four families of codes exceed those considered in \cite{codetables}. We have not found binary quantum codes with these lengths in the literature, thus we may conclude that they are new.

\subsubsection{Non-binary examples}
\label{nonbinary}

We devote this subsection to provide parameters of non-binary quantum error-correcting codes obtained with the same three procedures described in Subsection \ref{binary} for the binary case. Specifically, our codes come from applying either Theorem \ref{bueno}, or Theorem \ref{era2} and then Theorem \ref{Eldearxiv}, or Theorem \ref{Eldearxiv} applied to subfield-subcodes of codes given by Theorem \ref{bueno}. Most of them are new and we have not found other codes for comparison, but some of them can be compared and improve some codes in the recent literature.

With the triple $(q,n,b)=(3,2,4)$, applying Theorem \ref{era2} and then Theorem \ref{Eldearxiv}, after noticing that $A(q,t)=3$, we get ternary stabilizer quantum codes with parameters $[[72,64, \geq 3]]_3$, $[[72,60, \geq 4]]_3$ and $[[72,56, \geq 5]]_3$.

Consider now the triple $(q,n,b)=(5,2,6)$ and apply Theorem \ref{bueno} with $n'=1$. The value $D(q,t)$ equals 3 and we obtain a $5$-ary stabilizer quantum code with parameters $[[150, 136, \geq 5]]_5$ improving the $[[150,134, \geq 5]]_5$ code given in \cite{CaoCui}. With the help of Theorems \ref{era2} and \ref{Eldearxiv}, taking into account that $A(q,t)=11$, we also obtain new $5$-ary codes with length $n=300$ and remaining parameters as given in Table \ref{latablanb1}.

\begin{table}
\begin{center}
\begin{tabular}{||c|c|c|c|c|c|c|c|c|c|c|c||}
  \hline
  $k$ & 292 & 288 & 284 & 280 & 276 &  272 & 268 & 264 & 260 & 256 & 252 \\
   \hline
 $d\geq$ &  3 & 4 &  5 & 6 & 7 & 8 & 9 & 10 & 11 & 12& 13 \\
  \hline
\end{tabular}
\end{center}
\caption{Parameters of $5$-ary quantum codes of length 300}
\label{latablanb1}
\end{table}

Using now the triple $(5,2,26)$ and applying Theorem \ref{bueno} with $n'=1$, since $B(q,t)=16$, we get a family of stabilizer quantum codes with parameters $\left\{[[130,130-2(2i+1), \geq 2+i]]_5 \right\}_{i=1}^9$.

Moreover, considering the triple $(q,n,b)=(7,2,8)$ and applying Theorem \ref{bueno} with $n'=1$, $D(q,t)=5$ and we get $7$-ary quantum codes with parameters $[[392, 378, \geq 5]]_7$, $[[392, 374, \geq 6]]_7$ and $[[392, 370, \geq 7]]_7$, improving the codes with parameters $[[392, 376, \geq 5]]_7$, $[[392, 372, \geq 6]]_7$ and $[[392, 364, \geq 7]]_7$ given in \cite{CaoCui} and the code with parameters $[[392,368,\geq 7]]_7$ given in \cite{galher}. With the same triple, applying Theorems \ref{era2} and \ref{Eldearxiv}, since $A(q,t)=23$, we are able to obtain a family of $7$-ary quantum codes with parameters $[[784, 784 - 4 (i+1), \geq i+2]]_7$, $1 \leq i \leq 23$.

Finally, if we take $(q,n,b)=(7,2,50)$ and apply Theorem \ref{bueno} with $n'=1$, we get $B(q,t)=36$ and there is a family of stabilizer quantum codes with the following parameters: $\left\{ [[350, 350-2(2i+1), \geq i+2]]_7 \right\}_{i=1}^{15}$.



\subsubsection{Sporadic codes from trace-depending polynomials}
\label{sporadic}
In this subsection, we show that excellent quantum codes can be obtained by evaluating at the zeros of trace-depending polynomials. We consider here trace-depending polynomials which are different from those studied in this article, and some of our assertions are supported in calculations made with the computational algebra system Magma \cite{Magma}. It is an open question to develop a complete theory for studying this class of quantum error-correcting codes.

All the codes in this subsection are constructed as follows. Set $q=2$, $n=4$ and consider some new trace-depending polynomials $\mathcal{T}(X)$ different from the above considered $b$th trace-depending polynomials $\mathrm{Tr}_b(X)$. We have used \cite{Magma} to check that our polynomials $\mathcal{T}(X)$ have no multiple roots over the field $\mathbb{F}_{2^8}$. The number of roots of each $\mathcal{T}(X)$, say $m$, is not required to be the degree of $\mathcal{T}(X)$. Following the same notation and construction described in Subsection \ref{subfield}, we consider suitable sets $\Delta \subset \mathcal{H}$, codes $E_{\Delta, \mathcal{T}}$ obtained by evaluation under the map (\ref{evaluation}) -where $\mathrm{Tr}(X)$ is substituted by $\mathcal{T}(X)$- and subfield-subcodes $E_{\Delta, \mathcal{T}}|\mathbb{F}_{2^4}$ over the field $\mathbb{F}_{2^4}$. Proposition \ref{subf} determines bounds on the dimension and distance of these codes. Using \cite{Magma} again, we check that the codes $E_{\Delta, \mathcal{T}}|\mathbb{F}_{2^4}$ are Hermitian self-orthogonal. Finally, applying Theorem \ref{Eldearxiv} we get binary stabilizer quantum codes.

Table \ref{Records} shows polynomials $\mathcal{T}(X)$, sets $\Delta$ and parameters of the binary stabilizer quantum codes obtained, proving that, by selecting suitable trace-depending polynomials $\mathcal{T}(X)$, the above procedure produces records with respect to \cite{codetables}. Note that $a$ stands for a primitive element of the field $\mathbb{F}_{2^8}$ and $\mathrm{tr}(X):=\mathrm{tr}_{2n}(X)$ is the polynomial trace defined before Definition \ref{deftraza}. We conclude by explaining that the  sets $\mathfrak{C}^{'}_{i}$ that appear in Table \ref{Records} are some of the $16$-ary cyclotomic cosets over the set $\mathfrak{G} = \{0, 1, \ldots, 254\}$ considered in Subsection \ref{binary}.

\begin{table}[ht]
\centering
\begin{tabular}{||c|c|c|c||}
  \hline
  $\mathcal{T}(X)$ & $m$& $\Delta$ & Parameters $[[n, k, \geq d]]_2$\\ \hline \hline
  $1 + \mathrm{tr}(a^5 X^3)$ & $120$ & $\Delta_1= \cup_{i=0}^{5} \mathfrak{C}^{'}_{i}$ & $[[240,196, \geq 7]]_2 $\\ \hline
  $1 + \mathrm{tr}(a^5 X^3)$ & $120$ & $\Delta_2= \cup_{i=0}^{6} \mathfrak{C}^{'}_{i}$ & $[[240,188, \geq 8]]_2 $\\ \hline
  $1 + \mathrm{tr}(a^5 X^3)$ & $120$ & $\Delta_3=\cup_{i=0}^{7} \mathfrak{C}^{'}_{i}$ & $[[240,180, \geq 9]]_2 $\\ \hline
  $1 + \mathrm{tr}(a^5 X^3)$ & $120$ & $\Delta_4= \cup_{i=0}^{8} \mathfrak{C}^{'}_{i}$ & $[[240,172, \geq 10]]_2 $\\ \hline
  $1 + \mathrm{tr}(a^5 X^3)$ & $120$ & $\Delta_5= \cup_{i=0}^{9} \mathfrak{C}^{'}_{i}$ & $[[240,164, \geq 11]]_2 $\\ \hline
  $1 + \mathrm{tr}(a^5 X^3)$ & $120$ & $\Delta_6=\cup_{i=0}^{10} \mathfrak{C}^{'}_{i}$ & $[[240,156, \geq 12]]_2 $\\ \hline
  $1 + \mathrm{tr}(a^5 X^5)$ & $96$ & $\Delta_{7}=\cup_{i=0}^{11} \mathfrak{C}^{'}_{i}$ & $[[192,132, \geq 9]]_2 $\\ \hline
  $1 + \mathrm{tr}(a^5 X^5)$ & $96$ & $\Delta_{8}=\cup_{i=0}^{8} \mathfrak{C}^{'}_{i}$ & $[[192,124, \geq 10]]_2 $\\ \hline
  $1 + \mathrm{tr}(a X^{19} + X^{10})$ & $116$ & $\Delta_{9}= \cup_{i=0}^{7} \mathfrak{C}^{'}_{i}$ & $[[232,172, \geq 9]]_2 $\\ \hline
  $1 + \mathrm{tr}(a X^{19} + X^{10}) $ & $116$ & $\Delta_{10}= \cup_{i=0}^{8} \mathfrak{C}^{'}_{i}$ & $[[232,164, \geq 10]]_2 $\\ \hline
  $1 + \mathrm{tr}(a X^{19} + X^{10})$ & $116$ & $\Delta_{11}= \cup_{i=0}^{9} \mathfrak{C}^{'}_{i}$ & $[[232,156, \geq 11]]_2 $\\ \hline
  $1 + \mathrm{tr}(a X^{19} + X^{10})$ & $116$ & $\Delta_{12}=\cup_{i=0}^{10} \mathfrak{C}^{'}_{i}$ & $[[232,148, \geq 12]]_2 $\\ \hline
  $1 + \mathrm{tr}(a^3 X^{9} + X)$ & $112$ & $\Delta_{13}=\cup_{i=0}^{7} \mathfrak{C}^{'}_{i}$ & $[[224,164, \geq 9]]_2 $\\ \hline
  $1 + \mathrm{tr}(a^3 X^{9} + X)$ & $112$ & $\Delta_{14}=\cup_{i=0}^{8} \mathfrak{C}^{'}_{i}$ & $[[224,156, \geq 10]]_2 $\\ \hline
  $1 + \mathrm{tr}(a^3 X^{9} + X)$ & $112$ & $\Delta_{15}=\cup_{i=0}^{9} \mathfrak{C}^{'}_{i}$ & $[[224,148, \geq 11]]_2 $\\ \hline
  $1 + \mathrm{tr}(a^8 X^{25} + X^{10})$ & $100$ & $\Delta_{16}= \cup_{i=0}^{7} \mathfrak{C}^{'}_{i}$ & $[[200,140, \geq 9]]_2 $\\ \hline
  $1 + \mathrm{tr}(a^8 X^{25} + X^{10})$ & $112$ & $\Delta_{17}= \cup_{i=0}^{8} \mathfrak{C}^{'}_{i}$ & $[[200,132, \geq 10]]_2 $\\ \hline
  $1 + \mathrm{tr}(a^{17} X^{3} + X^{13})$ & $104$ & $\Delta_{18}= \cup_{i=0}^{7} \mathfrak{C}^{'}_{i}$ & $[[208,148, \geq 9]]_2 $\\ \hline
  $1 + \mathrm{tr}(a^{17} X^{3} + X^{13})$ & $104$ & $\Delta_{19}= \cup_{i=0}^{8} \mathfrak{C}^{'}_{i}$ & $[[208,140, \geq 10]]_2 $\\ \hline
  $1 + \mathrm{tr}(a^{17} X^{3} + X^{13})$ & $104$ & $\Delta_{20}= \cup_{i=0}^{9} \mathfrak{C}^{'}_{i}$ & $[[208,132, \geq 11]]_2 $\\ \hline
\end{tabular}
\caption{Binary quantum error-correcting records}
\label{Records}
\end{table}

\newpage

\bibliographystyle{plain}
\bibliography{biblioEA}

\end{document}